\documentclass[journal]{IEEEtran}

\usepackage[utf8]{inputenc}   
\usepackage[english]{babel}   

\usepackage{etoolbox}
\usepackage{nomencl}  
\makenomenclature 
\usepackage[hidelinks]{hyperref} 

\usepackage[linesnumbered,ruled]{algorithm2e}
\SetKw{Continue}{continue}

\usepackage{amsmath}          
\usepackage{amssymb}          
\usepackage{amsthm}           
\usepackage{bm}               
\usepackage{mathrsfs}         
\usepackage[mathscr]{euscript} 
\usepackage{bbm}              
\usepackage{dsfont}           
\usepackage{pifont}           

\usepackage{graphicx}         
\usepackage{epstopdf}         
\usepackage{epsfig}           
\usepackage{booktabs}         
\usepackage{multirow}
\usepackage{array}            
\usepackage{tabularx}         
\usepackage{blkarray}         
\usepackage{bigstrut}         
\usepackage{arydshln}         
\usepackage{subcaption}
\usepackage{makecell}

\usepackage{float}            
\usepackage{balance}          
\usepackage{microtype}        

\usepackage{xcolor}           
\usepackage{cite}             
\usepackage{enumitem}         
\usepackage{ulem}             
\newtheorem{theorem}{Theorem}
\newtheorem{lemma}{Lemma}

\makeatletter
\newcommand{\removelatexerror}{\let\@latex@error\@gobble}

\newcommand{\Rmnum}[1]{\expandafter\@slowromancap\romannumeral #1@}
\makeatother

\renewcommand\nomgroup[1]{%
  \item[\itshape
  \ifstrequal{#1}{A}{Abbreviations}{%
  \ifstrequal{#1}{P}{Power Grid Parameters}{%
  \ifstrequal{#1}{T}{Transformations}{%
  }}}%
]}

\addto\captionsenglish{\renewcommand{\figurename}{Fig.}}

\long\def\comment#1{}

\newfont{\bbb}{msbm10 scaled 700}

\newcommand{\RR}{\mathbb{ R}}

\newcommand{\av}{{\bf a}}

\newcommand{\cv}{{\bf c}}

\newcommand{\ev}{{\bf e}}

\newcommand{\hv}{{\bf h}}

\newcommand{\rv}{{\bf r}}

\newcommand{\xv}{{\bf x}}
\newcommand{\yv}{{\bf y}}
\newcommand{\zv}{{\bf z}}

\newcommand{\zerov}{{\bf 0}}

\newcommand{\Hm}{{\bf H}}
\newcommand{\Id}{{\bf I}}

\newcommand{\Km}{{\bf K}}

\newcommand{\Wm}{{\bf W}}

\newcommand{\Ec}{{\cal E}}

\newcommand{\Gc}{{\cal G}}
\newcommand{\Hc}{{\cal H}}

\newcommand{\Lc}{{\cal L}}

\newcommand{\Oc}{{\cal O}}

\newcommand{\Sc}{{\cal S}}

\newcommand{\Vc}{{\cal V}}

\newcommand{\RNum}[1]{\uppercase\expandafter{\romannumeral #1\relax}}

\renewcommand{\arg}{{\hbox{arg}}}

\newcommand{\eqdef}{\stackrel{\Delta}{=}}

\newcommand{\squeezeequ}{\medmuskip=2mu \thinmuskip=1mu \thickmuskip=3mu}

\newcommand{\supersqueezeequ}{\medmuskip=1mu \thinmuskip=0mu \thickmuskip=2mu \nulldelimiterspace=-1pt \scriptspace=0pt}

\newcommand{\Tsupersqueezeequ}{\medmuskip=0.1mu \thinmuskip=0mu \thickmuskip=0.1mu \nulldelimiterspace=-1pt \scriptspace=0pt}

\def\LRT#1#2{\!%
\raisebox{.2ex}{$%
\genfrac{}{}{0pt}{}%
  {\genfrac{}{}{0pt}{}{\scriptstyle\;#1}{\displaystyle\gtrless}}%
  {\raisebox{-1.25ex}{$\scriptstyle\;#2$}}%
$}%
\!}

\begin{document}
\title{BT-MTD: Bus Traversal-based Moving Target Defense for Smart Grid}
\author{
Jingyi Yan, Ke Sun,~\IEEEmembership{Senior~Member,~IEEE,} Zhenglin Li, Hongying Jia
\thanks{This work was supported by Natural Science Foundation of China (No. 62403301, 62421004, 62203289), and by Shanghai Pujiang Program (No. 23PJ1403000).

J. Yan and K. Sun are with the College of Computer Engineering and Science, Shanghai University, Shanghai, China (e-mail: jingyiyan@shu.edu.cn; ke\_sun@shu.edu.cn). (Corresponding author: Ke Sun.)

Z. Li is with the School of Future Technology, Shanghai University, Shanghai
(email: zhenglin\_li@shu.edu.cn).

Hongying Jia is with the Purple Mountain Laboratories, Nanjing 211111, China (e-mail: jiahongying@pmlabs.com.cn) 
}}
\maketitle

\begin{abstract}
Moving Target Defense (MTD) is a proactive security strategy designed to enhance cyber-resilience by dynamically altering system parameters, thereby preventing adversaries from acquiring the critical information needed to execute stealth attacks.
In this paper, we consider the case in which the operator modifies the admittance of branches to enable MTD, and focus on the problem of effectively protecting the system with fewer number of branch admittance modifications and shorter computational time. 
Specifically, we identify the ineffectual branches whose admittance modification do not contribute to the improvement of MTD effectiveness via theoretical analysis. 
Building on these insights, we propose the Bus Traversal-based MTD (BT-MTD), which is a bus-oriented algorithm that traverses over the buses of the network according to analytically derived guidelines. 
The performance of the BT-MTD is evaluated and compared with four existing strategies on standard IEEE test systems, demonstrating its robustness and superior performance in effectiveness, efficiency, and computational cost. 
The code of BT-MTD is available at: \href{https://github.com/YJY101/BT-MTD}{https://github.com/YJY101/BT-MTD}. 
\end{abstract}


\IEEEpeerreviewmaketitle

\section{Introduction}
\IEEEPARstart{T}{he} global energy landscape is undergoing a profound paradigm shift from conventional power grids to intelligent and highly integrated smart grids. 
At the core of this transformation is the deep fusion of information and communication technologies with the physical power infrastructure \cite{abdi_role-of-dl_2024, meyda_review_2024}.
The integration of advanced technologies, such as bidirectional communication, advanced sensing, and sophisticated control, unlocks the immense potential of the modern smart grid\cite{vasquez-plaza_aggregated_2025}. 
While these technologies enhance grid efficiency, reliability, and renewable energy integration capability, this progress is a double-edged sword. 
They also introduce new complexities and interdependencies\cite{sou_resilient_2024}.

This deep integration and the resulting complexities can create cybersecurity vulnerabilities, presenting significant threats to the stable and secure operation of the smart grid.
Among these, False Data Injection (FDI) attacks have emerged as a particularly severe threat \cite{ten_vulnerability_2008, Sun_Stealth_2020, mohammadpourfard_adaptive_2025}. 
An FDI attack is a sophisticated cyber-physical attack where an adversary strategically manipulates sensor measurements\cite{kececi_federated_2025} to mislead the smart grid's Energy Management System \cite{kosut_malicious_2011, liu_matrix-completion_2024}. 
The ultimate goal is to corrupt the system's state estimation\cite{esmalifalak_bad_2013}, potentially leading to flawed operational decisions, economic losses \cite{seshasai_design_2024, bi_profit-oriented_2022}
, or even cascading failures and widespread blackouts\cite{zhang_voltage-stability_2022, liao_cascading_2017}.
What makes these attacks especially pernicious is their ability to circumvent traditional Bad Data Detection (BDD) mechanisms, rendering them ``stealthy'' \cite{liu_false_2011, hug_vulnerability_2012}.
Recent studies even propose ``improved stealthy attacks'' designed to evade advanced detectors with accumulative properties by ensuring the attack signature is present for only a finite duration \cite{ning_improved_2024}.
The feasibility of achieving both deception and evasion critically depends on the attacker's knowledge of the smart grid's system information, for which the attackers can actively acquire the required system information by learning from historical or relevant information \cite{sun_2023_asymptotic, sun_2019_learning, lakshminarayana_2020_data, xu_blending_2023, chen_localization_2022, liu_matrix-completion_2024}. 
Moreover, the quasi-static operational nature of power systems \cite{cutsem_comprehensive_1995} exacerbates this issue, as historical information remains relevant over extended periods, giving adversaries a wide window to study the system and plan their attack.
As a result, active defense strategies, which proactively introduce controlled dynamics into the system, are essential for effectively invalidating an adversary's acquired information \cite{chen_localization_2022}.

As a representative of active defense strategies, Moving Target Defense (MTD) is a proactive security paradigm that creates a constantly changing attack surface by dynamically altering system configurations and parameters\cite{abdi_role-of-dl_2024, tan_survey_2023}.
In smart grids, this can be implemented by reconfiguring the network topology via transmission line switches \cite{zhang_2012_flexible, morrow_2012_topology, wang_2015_effects} or by modifying key electrical parameters, such as branch admittance, using Distributed Flexible AC Transmission System (D-FACTS) devices \cite{MTD, zhang_sliding_2023, kececi_federated_2025, bi_profit-oriented_2022, seshasai_design_2024}.
By compelling adversaries to continuously re-learn the system's state, these dynamic adaptations invalidate their prior knowledge, significantly increasing the complexity and cost of executing an attack, thereby strengthening the grid's overall security posture \cite{tan_strategy_2025}.

The effectiveness of MTD against FDI attacks is typically evaluated by its ability to eliminate or constrain the evasion capability of attacks.
Specifically, for the widely-used weighted least square (WLS) estimation and residual-based detection approach, the attacks that lie in the column space of the Jacobian matrix are undetectable stealth attacks \cite{liu_false_2011}.
Consequently, a key metric for MTD effectiveness is the dimension of the intersection between the column spaces of the Jacobian matrices before and after MTD deployment \cite{zhang_2020_analysis, zhang_2020_hiddenness, xu_2023_robust, liu_reactance_2018, tian_enhanced_2019}, which is mathematically equivalent to the rank of the matrices augmented by the pre-MTD and post-MTD Jacobian matrices.
This metric, known as the completeness of MTD, quantifies the residual stealth space available to attackers post-MTD.

Beyond effectiveness, the cost of implementing MTD is another critical perspective.
Costs encompass the deployment of D-FACTS devices \cite{lakshminarayana_2020_cost, zhang_double-benefit_2022} and the impact on power system operations, such as altered power flows \cite{wang_MMTD_2023, mohammadpourfard_adaptive_2025}.
Note that there is a fundamental tradeoff between cost and effectiveness, where achieving higher security typically requires a greater investment.
Addressing this trade-off has been a central challenge within MTD research.
Consequently, a primary focus of the research has been on optimizing the placement of D-FACTS devices to maximize security under a given deployment cost.
These placement strategies can be broadly categorized into random, heuristic, and graph-theoretic approaches.

Initial research explored using random perturbations to expose FDI attacks \cite{morrow_2012_topology, davis_power_2012, rahman_2014_moving}.
However, this ``blind" approach proved impractical for real-world applications due to its unpredictability, potential to disrupt grid stability, and frequent failure against adaptive attackers \cite{giraldo_decentralized_2022, xu_blending_2023, zhang_voltage-stability_2022, wang_small-signal_2025}.
To improve upon these random methods, research shifted to heuristic strategies.
These typically employ greedy or trial-and-error approaches to find modifications that decrease the intersection between the column spaces of the pre-MTD and post-MTD Jacobian matrices, which generally entails repeated computation of the rank of the augmented Jacobian matrix \cite{liu_reactance_2018, zhang_2020_analysis, tian_enhanced_2019, mohammadpourfard_adaptive_2025}.
A primary limitation of such approaches is the high computational complexity.
Evaluating the rank of an $m\times n$ matrix (via Singular Value Decomposition or QR factorization) typically requires $\mathcal{O}(n^3)$ (assuming $m >n$) per iteration \cite[Ch.~5, 8]{golub_matrix_2013}, which becomes prohibitive for large-scale grids.
Additionally, heuristic search processes are prone to entrapment in local optima, potentially failing to achieve the maximal security configuration \cite{mohammadpourfard_adaptive_2025}.
To overcome these limitations, recent work has turned to more structured, graph-theoretic approaches \cite{zhang_double-benefit_2022, li_feasibility_2020, lakshminarayana_moving_2021, liu_2021_optimal, xu_2023_robust, wang_stealthiness_2025}.
While much more efficient, they still face challenges, such as selecting redundant branches (increasing cost without improving security) or exhibiting sensitivity to measurement noise and initialization conditions.

In this paper, we consider the implementation of MTD through the modification of branch admittances using D-FACTS devices, focusing on ensuring MTD effectiveness with a reduced number of branch admittance modifications and minimal computational overhead.
The structure and main contributions of this work are summarized as follows.

(1) Firstly, we propose an optimization formulation where MTD effectiveness is quantified by the rank of an augmented Jacobian matrix.
Specifically, we demonstrate that removing the reference bus column from the Jacobian matrix does not affect the effectiveness analysis, see Lemma \ref{Lem:comH}. 

(2) Subsequently, through rigorous theoretical analysis, we identify three specific cases (Theorem \ref{Theo:deg1}-\ref{Theo:cycle}) involving redundant branches whose modification does not contribute to effectiveness improvement.
We precisely quantify the marginal rank gain when the operator perturbs an additional branch.
Furthermore, building on Theorem \ref{Theo:deg1}, we propose a novel graph-reduction procedure that iteratively eliminates leaf buses without sacrificing the rank, thereby shrinking the problem space and accelerating all subsequent computations.

(3) Based on these findings, we propose the Bus Traversal-Based MTD (BT-MTD) algorithm.
Unlike heuristic methods that rely on numerical rank evaluations, BT-MTD leverages discrete topological invariants (such as cycles and connectivity) to guide branch selection.
By utilizing constant-time pruning tests within a single traversal, the computational complexity is reduced to quasi-linear time. Finally, comparisons with four existing algorithms demonstrate the robustness and superior performance of BT-MTD.

\section{System Model}
\subsection{Notational Conventions}
Matrices are denoted by bold uppercase letters, such as $\Hm$, and vectors by bold lowercase letters, such as $\av$.
The $i$-th column of matrix $\Hm$ is represented by $\hv^i$.
The $\ell_2$ norm of a vector is represented by $\| \cdot \|_{2}$.
The rank and column space of a matrix are given by $\operatorname{rank}(\cdot)$ and $\operatorname{col}(\cdot)$, respectively.
Sets are represented by calligraphic font, such as $\Vc$, and their cardinality is given by $\mid \cdot \mid$; an exception is made for  the set of real numbers, which is denoted by $\RR$. 

The topology of a power system is described by a connected graph $\Gc = ( \Vc, \Ec)$, where $\Vc = \{ v_1, \ldots, v_{n+1}\}$ denotes the set of $n + 1$ buses (vertices), and $\Ec$ represents the set of $l$ branches (edges).
Each branch $e_{ij} \in \Ec$ connects two buses $v_i$ and $v_j$, and has an associated admittance $b_{ij}$.
The set of branches incident to bus $v_i$ is denoted by $\Ec_{v_i}$, and as a result, it holds for the degree of bus $v_i$ that $\deg(v_i) = \lvert \Ec_{v_i} \rvert$. 

\subsection{Linear Observation Model}
The linearized observation model is given by 
\begin{IEEEeqnarray}{c}
\label{Equ:LOM}
\yv = \bar{\Hm} \xv + \zv,
\end{IEEEeqnarray}
where $\yv \in \RR^{m}$ is the measurement vector; 
$\bar{\Hm} \in \RR^{m \times n}$ is the Jacobian matrix; 
$\xv \in \RR^{n}$ is the vector of state variables; 
and $\zv \in \RR^{m}$ is the additive measurement noise.
In typical settings, the measurements consist of the net active power injections at buses and the active power flows in both directions along transmission lines, while the state variables correspond to the voltage angle differences between buses and a chosen reference bus \cite{grainger_power_1994, abur_power_2004}.
Under this convention, the Jacobian matrix $\bar{\Hm}$ can be obtained by removing the column corresponding to the reference bus from a matrix $\Hm \in \RR^{m \times (n+1)}$, whose structure is characterized in the following lemma.

\begin{lemma}{\textnormal{\cite[Lemma 1 \& 2]{sun_2024_stealth}}}
\label{Lem:structureH}
The matrix $\Hm$ is given by 
\begin{IEEEeqnarray}{c}\label{Equ:add1}
\Hm = \sum_{e_{ij} \in \Ec } \! b_{ij} \Hm_{e_{ij}},
\end{IEEEeqnarray}
where $\Hm_{e_{ij}}$ is a matrix that only the $i$-th and $j$-th columns are nonzero vectors with relation $\hv_{e_{ij}}^{i} = -\hv_{e_{ij}}^{j}$, and $\hv_{e_{ij}}^{i}$ has only four nonzero elements and is given by 
\begin{IEEEeqnarray}{c}
\begin{array}{lllll}
\hv_{e_{ij}}^{i} = [ \ldots, &1, \ldots, &-1, \ldots, &1,\ldots,& -1, \ldots]^{\sf T}\\
        &\mbox{\footnotesize i-th}& \mbox{\footnotesize j-th} &\mbox{\footnotesize (n+1+k)-th}            &\mbox{\footnotesize (n+1+l+k)-th}
\end{array}
\Tsupersqueezeequ
\IEEEeqnarraynumspace
\end{IEEEeqnarray}
with $k$ denoting the index of branch $e_{ij}$. 
Furthermore, it holds that $\operatorname{rank}(\Hm) = n$. 
\end{lemma}

Lemma \ref{Lem:structureH} shows that matrix $\Hm$ can be decomposed into a weighted sum of some base matrices, where each base matrix corresponds to a specific branch $e_{ij}$ and has a sparse structure. 
However, the rank of $\Hm$ is $n$, which implies that it is a rank-deficient matrix. 
Conversely, the matrix $\bar{\Hm}$ is a full-rank matrix, since it is formed by removing a column from $\Hm$. 
Note that $\Hm$ can also be considered as a Jacobian matrix. 
Specifically, under the full Jacobian matrix $\Hm$, the state variables represent the absolute voltage angles at all $n+1$ buses; 
whereas under the reduced Jacobian matrix $\bar{\Hm}$, the state variables represent the voltage angle differences between a chosen reference bus and each of the other $n$ buses. 
To distinguish them, $\Hm$ is referred as ``{\it full Jacobian matrix}'' and $\bar{\Hm}$ is referred as ``{\it reduced Jacobian matrix}'' in the remaining of the paper.

A direct consequence of the matrix decomposition in Lemma \ref{Lem:structureH} is that each column of the full Jacobian matrix $\Hm$ can also be decomposed, which is stated in the following lemma. 
\begin{lemma}
\label{Lem:colH}
The $i$-th column $\hv^{i}$ of the full Jacobian matrix $\Hm$ is given by
\begin{IEEEeqnarray}{c}
\label{Equ:colH}
\hv^{i} = \!\! \sum_{e_{ij} \in \Ec_{v_i}} \!\!\! b_{ij} \hv_{e_{ij}}^{i}.
\end{IEEEeqnarray}
\end{lemma}
\begin{proof}
The proof is presented in Appendix \ref{Pro:Lem:colH}.
\end{proof}

\subsection{State Estimation and Abnormality Detection}
State estimation refers to the process of estimating the state variables from noisy measurements, for which the WLS estimator is commonly adopted. 
The closed-form expression of the WLS estimator is given by 
\begin{IEEEeqnarray}{rCl}
\hat{\xv} =
 \arg \  \min_{\ev \in \RR^{n}} \| \Wm \left(\yv - \bar{\Hm} \ev \right)\|_{2}^2 =\left( \bar{\Hm}^{\sf T} \Wm \bar{\Hm} \right)^{-\!1}\bar{\Hm}^{\sf T}\Wm \yv, \IEEEeqnarraynumspace \squeezeequ \label{Equ:WLS0} 
\end{IEEEeqnarray}
in which $\Wm$ is a diagonal weighting matrix whose entries are typically the inverse of the noise variances \cite[(8.16)]{kay_1993_fundamentals}. 
The corresponding estimation residual $r(\yv)$ is calculated as
\begin{IEEEeqnarray}{c}\label{Equ:res}
r(\yv) =  \| \yv - \bar{\Hm} \hat{\xv} \|_{2}^2.
\end{IEEEeqnarray}
The full-rank property of the reduced Jacobian matrix $\bar{\Hm}$ guarantees that $\hat{\xv} $ in \eqref{Equ:WLS0} is unique for a given measurement vector $\yv$. 

A conventional method for detecting system anomalies is the BDD mechanism.
This procedure is a hypothesis testing problem that compares the estimation residual $r(\yv)$ against a predetermined threshold $\tau$, i.e. 
\begin{IEEEeqnarray}{c}
r(\yv) \LRT{\Hc_1}{\Hc_0} \tau, \label{Equ:ResTest}
\end{IEEEeqnarray}
in which the null hypothesis $\Hc_0$ states that no bad data is present, the alternative hypothesis $\Hc_1$ states that bad data is present, and $\tau$ is the threshold balancing the probabilities of false alarm and detection \cite{abur_power_2004}.

\subsection{False Data Injection Attacks}
FDI attacks disrupt the power system operation by introducing malicious data into acquired measurements.
For the observation model in \eqref{Equ:LOM}, the measurement model under attacks is given by 
\vspace{-0.5em}
\begin{IEEEeqnarray}{c}
\yv_a = \bar{\Hm} \xv + \zv +\av, 
\end{IEEEeqnarray}
where $\yv_a \in \RR^{m} $ is the compromised measurement vector and $\av \in \RR^{m}$ is the attack vector.
A stealth attack is one that achieves the dual objectives of (i) deviating the estimated state from the original value, while (ii) evading the BDD.
It is shown in \cite{liu_false_2011} that the attack vector that lies within the column space of $\bar{\Hm}$, i.e. $\av \in \operatorname{col}(\bar{\Hm})$, is a stealth attack. 

\section{Problem Formulation}
\subsection{Problem Statement}
MTD is a proactive strategy that boosts cyber resilience by dynamically changing system parameters, depriving attackers of the critical information needed for stealth attacks. 
This study focuses on an MTD approach that uses D-FACTS devices to alter branch admittances, rendering the attacker's knowledge of the Jacobian matrix inaccurate \cite{rogers_2008_some}.


Initially, both the operator and attacker know the original reduced Jacobian matrix $\bar{\Hm}$.
Subsequently, the operator modifies the admittance of some selected admittances, resulting in an updated Jacobian matrix $\bar{\Hm}^\prime \in \RR^{m \times n}$.
However, the attacker is assumed to be unaware of this, and continues to construct the attack vector as $\av = \bar{\Hm}\cv$, based on the outdated Jacobian matrix. 
Consequently, if this vector does not lie in the new column space (i.e., $\av \notin \operatorname{col}(\bar{\Hm}^\prime)$), the attack becomes detectable.

Hence, the effectiveness of the MTD strategy is measured by the reduction of the {\it stealthy attack space} $\Sc_a$, which is defined as the intersection of the column spaces before and after the MTD alteration, i.e. 
\begin{IEEEeqnarray}{c}
\Sc_a = \operatorname{col}(\bar{\Hm}) \cap \operatorname{col}(\bar{\Hm}^{\prime}).
\end{IEEEeqnarray}
Accordingly, the objective of the operator is find an MTD strategy (or equivalently, a modified Jacobian matrix $\bar{\Hm}^{\prime}$) to minimize the dimension of the stealthy attack space, i.e.
\begin{IEEEeqnarray}{c}
\label{Equ:ObjFunMin}
\min_{\bar{\Hm}^{\prime}} \ \dim (\Sc_a).
\end{IEEEeqnarray}

\subsection{Reformulation and Methodological Considerations}
\label{SubSec:RMC}
The optimization problem in \eqref{Equ:ObjFunMin} is in general intractable. 
Although the reduced Jacobian matrix $\bar{\Hm}$ is a full-rank matrix, it lacks the structure regularity of the full Jacobian $\Hm$, see Lemma \ref{Lem:structureH}.
To facilitate a more tractable formulation, the following lemma proves that the dimension of the stealthy attack space is invariant whether the analysis is conducted using the full Jacobian matrix $\Hm$ or the reduced Jacobian matrix $\bar{\Hm}$.

\begin{lemma}
\label{Lem:comH}
Let $\Hm^\prime$ and $\bar{\Hm}^\prime$ be the post-MTD Jacobian matrix of $\Hm$ and $\bar{\Hm}$, respectively. 
For the same MTD strategy, it holds that
\begin{IEEEeqnarray}{c}\label{Equ:comH}
\operatorname{col}(\bar{\Hm}) \cap \operatorname{col}(\bar{\Hm}^{\prime})
=
\operatorname{col}(\Hm) \cap \operatorname{col}(\Hm^{\prime}).
\end{IEEEeqnarray}
\end{lemma}
\begin{proof}
The proof is presented in Appendix \ref{Pro:Lem:comH}.
\end{proof}

Note that the same MTD strategy implies the set of branches whose admittance are changed and the changed admittance values are all the same. 
Lemma \ref{Lem:comH} establishes that, for the given  optimization problem, analyzing the reduced Jacobian matrix $\bar{\Hm}$ is equivalent to analyzing the full Jacobian matrix $\Hm$. 
Consequently, {\it the full Jacobian matrix $\Hm$ is used for all subsequent analysis}, as its structural properties are more amenable to theoretical investigation. 
Specifically, the impact of the reference bus can be ignored, and it is unnecessary to consider whether a branch is connected to the reference bus, which simplifies the following
analysis.


Furthermore, the detectability of the attacks does not change when either Jacobian matrix is considered.
Under WLS-based BDD, attack vector $\av$ is undetectable if $\av \in \operatorname{col}(\bar{\Hm}^{\prime})$, since such vectors leave the measurement residual $\rv = (\Id - \Km)\yv$ unchanged, where $\Km = \bar{\Hm}^{\prime}(\bar{\Hm}^{\prime {\sf T}} \Wm \bar{\Hm}^{\prime})^{-\!1}\bar{\Hm}^{\prime {\sf T}} \Wm$ is the oblique projection onto $\operatorname{col}(\bar{\Hm}^{\prime})$ \cite[Ch.~2]{abur_power_2004}.
Moreover, since $\operatorname{col}(\bar{\Hm}^{\prime}) = \operatorname{col}(\Hm^{\prime})$ (see the proof of Lemma \ref{Lem:comH}), the WLS detectability criterion is identical whether formulated via $\Hm$ or $\bar{\Hm}$.

\begin{lemma}
\label{Lem:deltaH}
Let $\Ec_{m}$ be the set of branches whose admittances are modified by MTD.
And let $\delta_{ij} = b_{ij}^\prime / b_{ij}$ be the ratio of this modification for branch $e_{ij}$ with $b_{ij}^\prime$ denoting the post-MTD admittance of branch  $e_{ij}$.
The mismatch between the post- and pre-MTD Jacobian matrix is given by 
\begin{IEEEeqnarray}{c}\label{Equ:add0}
\Delta \Hm = \Hm^{\prime} - \Hm = \!\! \sum_{e_{ij} \in \Ec_{m}} \!\!\! (\delta_{ij} \! - \! 1 ) b_{ij} \Hm_{e_{ij}}.
\end{IEEEeqnarray}
\end{lemma}
\begin{proof}
The proof is provided in Appendix \ref{Pro:Lem:deltaH}.
\end{proof}

In the remaining of this paper, we assume that $\delta_{ij} \neq 1, \forall e_{ij} \in \Ec_{m}$, which implies that the admittance of any branch included in $\Ec_m$ is altered by the MTD strategy, i.e. its post-MTD value differs from the original.
Using the result in Lemma \ref{Lem:deltaH}, the stealthy attack space minimization problem is reformulated into an equivalent rank-maximization problem. 
\begin{lemma}
\label{Lem:ObjFunMax}
The dimension of the stealthy attack space is
\begin{IEEEeqnarray}{c}
\dim(\Sc_a) = 2n - \operatorname{rank} \left( \begin{bmatrix} \Hm \mid \Delta \Hm \end{bmatrix} \right).
\end{IEEEeqnarray}
As a result, the optimization problem in \eqref{Equ:ObjFunMin} is equivalent to 
\begin{IEEEeqnarray}{c}
\label{Equ:OptPro}
\max_{\Delta \Hm} \
\operatorname{rank}
\left(
\begin{bmatrix}
\Hm \mid \Delta \Hm
\end{bmatrix}
\right).
\end{IEEEeqnarray}
\end{lemma}
\begin{proof}
The proof is provided in Appendix \ref{Pro:OptPro}.
\end{proof}

Lemma \ref{Lem:ObjFunMax} shows that minimizing the dimension of the stealthy attack space is equivalent to maximizing the rank of the augmented matrix $\left[ \Hm \mid \Delta \Hm \right]$. 
Furthermore, from Lemma \ref{Lem:deltaH}, it can be found that $\Delta \Hm$ is composed of $|\Ec_m|$ matrices, but $\Hm^{\prime}$ is composed of $|\Ec|$ matrices. 
To that end, considering $\Delta \Hm$ simplifies the characterization of the effectiveness of MTD. 

However, only optimizing the objective function in \eqref{Equ:OptPro} is insufficient, in which the number of branches whose admittances are modified must also be taken into account.  
If two MTD strategies achieve the same value of $\operatorname{rank} \left(\left[ \Hm \mid \Delta \Hm \right]\right)$, the one with a smaller value of $|\Ec_{m}|$ is preferable, as it requires fewer D-FACTS devices.
To incorporate this, the optimization problem in \eqref{Equ:OptPro} is further developed into the bi-objective optimization problem
\begin{IEEEeqnarray}{c}
\label{Equ:BiObjectiveOpt}
 \max_{\Ec_m \subseteq \Ec, \  \bm{\delta} \in \{\mathbb{R}^{|\Ec_m|}/\zerov \}}  \left( \operatorname{rank}\left( \begin{bmatrix} \Hm \mid \Delta \Hm \end{bmatrix} \right), -|\Ec_m| \right), \ 
 \textnormal{s.t.}  \ \eqref{Equ:add0}, \ \IEEEeqnarraynumspace \supersqueezeequ 
\end{IEEEeqnarray}
in which $\bm{\delta} = [ \ldots, \delta_{ij}, \ldots]^{\sf T}$ is the vector of branch admittance modification ratios defined in Lemma \ref{Lem:deltaH}.

\section{Theoretical Characterization}

As a combinatorial subset selection problem, the optimization in \eqref{Equ:BiObjectiveOpt} is intrinsically NP-hard \cite{liu_2021_optimal} and computationally challenging to solve. 
Moreover, previous studies indicate that protecting certain branches fails to enhance MTD effectiveness \cite{zhang_2020_analysis,liu_reactance_2018}.
To that end, we aim to provide theoretical support for such situations in this section, which is crucial for minimizing $|\Ec_m|$ in \eqref{Equ:BiObjectiveOpt} without decreasing $\operatorname{rank}\left( \begin{bmatrix} \Hm \mid \Delta \Hm \end{bmatrix} \right)$. 

The following theorem establishes that modifying the admittance of branches connected to degree-one (leaf) buses is a redundant action \emph{for the purpose of rank maximization}\footnote{The redundancy here is characterized with respect to rank maximization only; once an actuation-cost function is included in the objective, this redundancy may no longer apply, since a leaf-incident branch may still be operationally preferred when its actuation cost is materially lower than that of interior branches.}.
\begin{theorem}
\label{Theo:deg1}
Let $\Ec_{leaf}$ be the set of all branches incident to degree-one buses. 
Then for any set $\Ec_m \cup \Ec_{add}$ with $\Ec_m$ in Lemma \ref{Lem:deltaH} and $\Ec_{add} \subseteq \Ec_{leaf}$, it holds that 
\begin{IEEEeqnarray}{c} 
\operatorname{rank} 
\left( 
\begin{bmatrix} 
\Hm \mid \Delta \Hm^{\prime} 
\end{bmatrix} 
\right)
=
\operatorname{rank}
\left(
\begin{bmatrix}
\Hm \mid \Delta \Hm
\end{bmatrix}
\right),
\end{IEEEeqnarray}
where $\Delta\Hm$ and $\Delta \Hm^{\prime} $ are the mismatch between the post- and pre-MTD Jacobian matrix when $\Ec_m$ and $\Ec_m \cup \Ec_{add}$ are protected by MTD, respectively. 
\end{theorem}
\begin{proof}
The proof is provided in Appendix \ref{Pro:Theo:deg1}. 
\end{proof}
Theorem \ref{Theo:deg1} demonstrates that modifying the admittance of branches connecting to leaf buses cannot increase the rank of the augmented matrix. 
As a result, to minimize $|\Ec_m|$ in \eqref{Equ:BiObjectiveOpt} without decreasing $\operatorname{rank}\left( \begin{bmatrix} \Hm \mid \Delta \Hm \end{bmatrix} \right)$, such branches should not be considered. 

This insight motivates a graph reduction procedure to systematically eliminate ineffectual branches and identify the core topology relevant to the optimization. 
Note that the graph reduction procedure is iterative, in which  removing a degree-one bus along with its incident branch may create new degree-one buses.
The final output of this procedure is a reduced core graph $\Gc^{\prime} = (\Vc^{\prime}, \Ec^{\prime})$, where each bus $v_i \in \Vc^{\prime}$ has a degree strictly greater than one and an updated set of incident branches $\Ec_{v_i}^{\prime}$. 
The following theorem applies to this reduced graph.
 
\begin{theorem}
\label{Theo:deg2}



In the reduced graph $\Gc^{\prime} = (\Vc^{\prime}, \Ec^{\prime})$, let $\Ec_{v_i}^{\prime\prime} = \{\Ec_{v_i}' \cap \Ec_m^{\sf C} \}$ be the set of branches that are incident to bus $v_i \in \Vc^{\prime}$ but not included in $\Ec_m$.
It holds that 
\begin{IEEEeqnarray}{c}
\operatorname{rank}
\left(
\begin{bmatrix}
\Hm \mid \Delta \Hm^{\prime}
\end{bmatrix} 
\right) 
- 
\operatorname{rank} 
\left( 
\begin{bmatrix} 
\Hm \mid \Delta \Hm 
\end{bmatrix} 
\right) \nonumber \\
=
\begin{cases}
\lvert \Ec_{add} \rvert, & \textnormal{if } \Ec_{add} \subset \Ec_{v_i}^{\prime\prime}, \\
\lvert \Ec_{add} \rvert - 1, & \textnormal{if } \Ec_{add} = \Ec_{v_i}^{\prime\prime},
\end{cases}
\end{IEEEeqnarray}
where $\Delta\Hm$ and $\Delta \Hm^{\prime} $ are the mismatch between the post- and pre-MTD Jacobian matrix when $\Ec_m$ and $\Ec_m \cup \Ec_{add}$ are protected by MTD, respectively. 

\end{theorem}
\begin{proof}
The proof is provided in Appendix \ref{Pro:Theo:deg2}.
\end{proof}

Theorem \ref{Theo:deg2} shows that for all branches that are incident to a bus, as long as there exists a branch whose admittance remains unmodified, further modifying the admittance of a branch exactly increases the rank by one.
However, when there is only one branch whose admittance is unmodified, further modifying the admittance of that branch does not increase the rank. 
As a result, to minimize $|\Ec_m|$ in \eqref{Equ:BiObjectiveOpt} without decreasing $\operatorname{rank}\left( \begin{bmatrix} \Hm \mid \Delta \Hm \end{bmatrix} \right)$, an optimal strategy is modifying at most $\lvert \Ec_{v_i}^\prime \rvert - 1$ branches incident to any single bus.



The following theorem considers the case in which adding a branch to the set $\Ec_{m}$ forms a cycle.


\begin{theorem}
\label{Theo:cycle}
Let ${\cal E}_{m}^\prime = {\cal E}_{m} \cup \{ e_{ij} \}$, where $e_{ij}$ forms a cycle with a subset of ${\cal E}_{m}$ in $\Gc'$. 
It holds that
\begin{IEEEeqnarray}{c}
\operatorname{rank}
\left(
\begin{bmatrix}
\Hm \mid \Delta \Hm^\prime
\end{bmatrix} 
\right) 
=
\operatorname{rank}
\left(
\begin{bmatrix}
\Hm \mid \Delta \Hm 
\end{bmatrix} 
\right) 
\end{IEEEeqnarray}
where $\Delta\Hm$ and $\Delta \Hm^{\prime} $ are the mismatch between the post- and pre- MTD Jacobian matrix when $\Ec_m$ and $\Ec_{m} \cup \{ e_{ij} \} $ are protected by MTD, respectively. 
\end{theorem}
\begin{proof}
The proof is provided in Appendix \ref{Pro:Theo:cycle}.
\end{proof}

Theorem \ref{Theo:cycle} demonstrates that if adding a branch to the set $\Ec_{m}$ forms a cycle with already-selected branches, there is no need to add that branch, since it has no contribution in increasing the rank objective in \eqref{Equ:BiObjectiveOpt}. 
Consequently, to achieve an effective MTD with smaller value of $\lvert \Ec_{m} \rvert$ (i.e. modify the admittance of less branches), 
the set of branches whose admittances are changed should form a forest (collection of trees) in graph $\Gc$.

\vspace{-0.5em}

\section{BT-MTD Algorithm}

Building on previous findings, we introduce the BT-MTD algorithm (Algorithm \ref{Algo:BT-MTD}) to minimize the stealthy attack space dimension $\dim (\Sc_a)$, using fewer branch admittance modifications, i.e. with smaller value of $\lvert \Ec_{m} \rvert$.
This bus-oriented algorithm starts on a simplified graph devoid of degree-one buses (Theorem \ref{Theo:deg1}). 
From a random initial bus, it traverses the network topology, prioritizing lower-degree buses and the branches that link them.


Note that the previous results in \cite{xu_2023_robust}, \cite[Section IV.C]{zhang_2020_analysis}, \cite{liu_2020_optimal}, and our results in Section IV show that the topology is the core factor affecting the effectiveness of MTD, and that this characterization holds for all post-MTD admittance values that differ from the pre-MTD ones.
To that end, the focus of BT-MTD is to identify the branches whose admittances should be changed.
However, it is important to clarify that this independence from admittance values assumes a \emph{perfect stealth attack}, which leaves the residual entirely unchanged.
In practice, if the applied admittance perturbations are too small, the resulting increase in the BDD residual may not exceed the detection threshold $\tau$ in \eqref{Equ:ResTest}, potentially causing the BDD to miss the attack.

The notations employed in Algorithm \ref{Algo:BT-MTD} are defined as follows.
$\Vc_m$ denotes the candidate buses for the next iteration;
$\Ec_c$ is the set of branches that is already covered during the traversal;
and $ {\cal L}_{v_i}$ is a list of elements belonging to $\Ec_{v_i}$ with a defined order.
Note that $\Ec_c$ is different from $\Ec_m$, in which it holds that $\Ec_m \subseteq \Ec_c$.
Specifically, if adding a branch into the set $\Ec_m$ forms a cycle, or adding a branch into the set $\Ec_c$ makes $\Ec_c$ includes all the branches incident to a bus, then it is known from Theorem \ref{Theo:cycle} or Theorem \ref{Theo:deg2} that this branch should not be included in $\Ec_m$ and does not need to be considered in the following iterations.
To that end, it is only added into the set $\Ec_c$ and considered as ``covered'' by BT-MTD.
Furthermore,
a Boolean variable \texttt{flag} is utilized within the \texttt{ProcessDegreeOne} function to control its behavior.
When flag is false, the function executes the initial graph reduction (as discussed after Theorem 1).
Conversely, when flag is true, the function operates within the main loop, dynamically updating $\Vc_m$ whenever a bus is pruned.

\begin{algorithm}[t!]
\caption{BT-MTD Algorithm}
\label{Algo:BT-MTD}
\SetKw{KwStatex}{}
\KwIn{Grid topology ${\cal G} = (\Vc, \Ec)$}
\KwOut{Set of branches with modified admittance $\Ec_{m}$}
\SetKwProg{Fn}{Function}{:}{}
\SetKwFunction{BTMTD}{BTMTD}
\SetKwFunction{ProcessDegreeOne}{ProcessDegreeOne}
\DontPrintSemicolon
\Fn{\BTMTD{${\cal G}$}}{
    Initialize $\Ec_{m} \leftarrow \varnothing$, $\Vc_{m} \leftarrow \varnothing$, $\Ec_{c} \leftarrow \varnothing$ \;
    \tcc{prune leaf buses (Theorem~\ref{Theo:deg1})}
    \ForEach{$v_i \in \Vc$}{
        \ProcessDegreeOne{$v_i$, \texttt{false}} \;
    }
    \tcc{initialize the traversal}
    Select a random bus $v_i \in \Vc$ \;
    $\Vc_{m} \leftarrow \Vc_{m} \cup \{ v_i \}$ \;
    \tcc{traverse to build forest}
    \While{$\lvert \Ec_{c} \rvert < \lvert \Ec \rvert$}{
        Select $v_i \in \Vc_{m}$ with the minimum degree \;
        $\Vc_{m} \leftarrow \Vc_{m} \setminus \{v_i\}$ \;
        Get sorted list $\Lc_{v_i}$ for each $e_{ij} \in \Ec_{v_i}$ in ascending order of $\lvert \Ec_{v_j} \rvert$ \;
        \ForEach{$e_{ij} \in \Lc_{v_i}$}{
            \If{\texttt{not} \texttt{FormsCycle}($e_{ij}$, $\Ec_m$)}{
                $\Ec_{m} \leftarrow \Ec_{m} \cup \{e_{ij}\}$ \;
                $\Ec_{c} \leftarrow \Ec_{c} \cup \{e_{ij}\}$ \;
                $\Vc_{m} \leftarrow \Vc_{m} \cup \{v_j\}$ \;
                $\Ec_{v_i} \leftarrow \Ec_{v_i} \setminus \{e_{ij} \}$ \;
                $\Ec_{v_j} \leftarrow \Ec_{v_j} \setminus \{e_{ij} \}$ \;
                \tcc{prune leaves after updates}
                \ProcessDegreeOne{$v_j$, \texttt{true}} \;
                \ProcessDegreeOne{$v_i$, \texttt{true}} \;
            }
        }
    }
}
\vspace{0.5em}
\Fn{\ProcessDegreeOne{$v_i$, \texttt{flag}}}{
    \tcc*[r]{\texttt{flag}=\textbf{true}: maintain $\Vc_m$ in main loops; \texttt{false}: graph reduction}
    \While{$\lvert \Ec_{v_i}\rvert = 1$}{
        $\Vc \leftarrow \Vc \setminus \{v_i\}$ \;
        Let $e_{ij}$ be the sole branch in $\Ec_{v_i}$ \;
        \If{\texttt{flag}}{
            $\Vc_{m} \leftarrow \Vc_{m} \setminus \{v_i\}$ \;
            $\Vc_{m} \leftarrow \Vc_{m} \cup \{v_j\}$ \;
        }
        $\Ec_{c} \leftarrow \Ec_{c} \cup \{ e_{ij} \}$ \;
        $\Ec_{v_i} \leftarrow \Ec_{v_i} \setminus \{e_{ij} \}$ \;
        $\Ec_{v_j} \leftarrow \Ec_{v_j} \setminus \{e_{ij} \}$ \;
        $v_i \leftarrow v_j$ \;
    }
}
\end{algorithm}

\subsection{Algorithm Description}
The BT-MTD algorithm is mainly composed of two phases and a defined function.

\subsubsection{Initialization and Pre-processing (Lines 2-7)}
This phase begins by initializing the set of branches with modified admittance $\Ec_{m}$, candidate buses $\Vc_{m}$, and covered branches $\Ec_{c}$ as empty (Line 2).
Subsequently, a graph pre-processing step (Lines 3-5, guided by Theorem \ref{Theo:deg1}) simplifies the network topology, in which the \texttt{ProcessDegreeOne} function is repeatably called to identify and eliminate the degree-one buses and their incident branches until there is no degree-one bus remains.
Then, a random bus is selected as the starting bus (Line 6) and added to the candidate bus set $\Vc_m$ (Line 7).

\subsubsection{Main Iterative Loop (Lines 8-21)}
Following the first phase, the second phase commences by iteratively adding branches to set $\Ec_m$ until all branches are covered (i.e. $ \Ec_{c} = \Ec$ ).
Specifically, in each iteration,
\begin{itemize}
    \item the algorithm selects the bus with the minimum degree from the candidate set $\Vc_{m}$, and removes it from  $\Vc_{m}$ (Lines 9-10);
    \item It then examines the branches connected to this bus, and prioritizes adding the branches that lead to low-degree neighbors into the set $\Ec_m$ (Lines 11-12).
    A branch is only added to the set $\Ec_m$, if it doesn't create a cycle, which is a core rule from Theorem \ref{Theo:cycle} (Line 13).
\end{itemize}

After a valid branch $e_{ij}$ is added to $\Ec_m$ and $\Ec_c$ (Lines 14-15), the algorithm updates for the next iteration.
\begin{itemize}
    \item The neighboring bus $v_j$ is added to the candidate set $\Vc_{m}$ (Line 16).
    \item The branch $e_{ij}$ is removed from both $\Ec_{v_i}$ and $\Ec_{v_j}$ (Lines 17-18), which will also makes the degree of both bus $v_j$ and $v_i$ decreasing by one.
    \item Finally, the \texttt{ProcessDegreeOne} function is called on both bus $v_j$ and bus $v_i$ (Lines 19-20), but for different strategic reasons.
    For the neighbor bus $v_j$, the function checks if its degree has now decreased to one.
    If it has, Theorem \ref{Theo:deg1} indicates that the bus $v_j$ offers no more strategic value, and its last connected branch is simply considered as ``covered''.
    For the original bus $v_i$, the function checks if there is only one branch remains unmodified.
    If it is true, Theorem \ref{Theo:deg2} indicates that modifying the last remaining branch does not decrease the dimension of the stealthy attack space; and that branch can be simply considered as ``covered''.
\end{itemize}

The aforementioned two phases strategically simplify the graph topology, avoiding changing the admittance of branches that do not contribute to the rank objective function in \eqref{Equ:BiObjectiveOpt}.

\subsubsection{Function \texttt{ProcessDegreeOne} (Lines 24-36)}
This function executes a ``depth-first search'' only when the degree of a bus is one or is decreased to one.
When the flag is ``false'', the function implements the graph pre-processing and simplifies the network topology.
Specifically, when the degree of bus $v_i$ is one or is decreased to one (Line 25),
it removes $v_i$ from the set of $\Vc$ (Line 26) and adds its sole connecting branch $e_{ij}$ to the covered set $\Ec_c$ (Line 32).
Also the branch is removed from the adjacency sets of both $v_i$ and $v_j$ (Lines 33-34).
The traversal subsequently advances to the neighboring bus $v_j$, which becomes the new $v_i$ for the next iteration (Line 35).
This procedure finishes when a bus with a degree greater than one is encountered.

When the flag is ``true'', the function implements the graph update in the main iterative loop, in which the set of candidate buses $\Vc_m$ also needs to be updated.
Specifically, when the degree of bus $v_i$ is decreased to one, there is no need to explore that bus in the following iteration, and it is removed from the set $\Vc_m$ (Lines 29).
And bus $v_j$ is added into the candidate bus set $\Vc_m$ (Line 30).

{\bf Remark (Pareto Optimality):}
The termination condition $|\Ec_c| < |\Ec|$ specified in Line 8 ensures that the algorithm exhaustively explores the entire graph, thereby guaranteeing the maximum attainable rank by processing all branches.
However, this condition can be generalized to accommodate real-world resource limitations by substituting it with $|\Ec_m| \leq k$, where $k$ is the maximum number of branches that can be protected.
Varying $k$ allows the algorithm to explore the inherent trade-off in \eqref{Equ:BiObjectiveOpt}: iterating over different $k$ values generates a Pareto frontier, where each point represents a distinct balance between protection cost and the achieved rank.

{\bf Remark (Restricted D-FACTS Placement):}
Algorithm~\ref{Algo:BT-MTD} can be adapted to handle hard candidate restrictions, such as when only a specific subset of branches $\Ec_{\text{cand}}\subseteq\Ec$ is actuable.
The selection of actuable branches can be constrained simply by replacing the per-bus list $\mathcal{L}_{v_i}$ in Line 12 with $\mathcal{L}_{v_i}\cap\mathcal{E}_{\text{cand}}$.
The conclusions established in Theorems 1 through 3 remain valid, as they hold for any arbitrary subset of  $\mathcal{E}$.

\begin{figure}[t]
\centering
\includegraphics[width=0.7\linewidth]{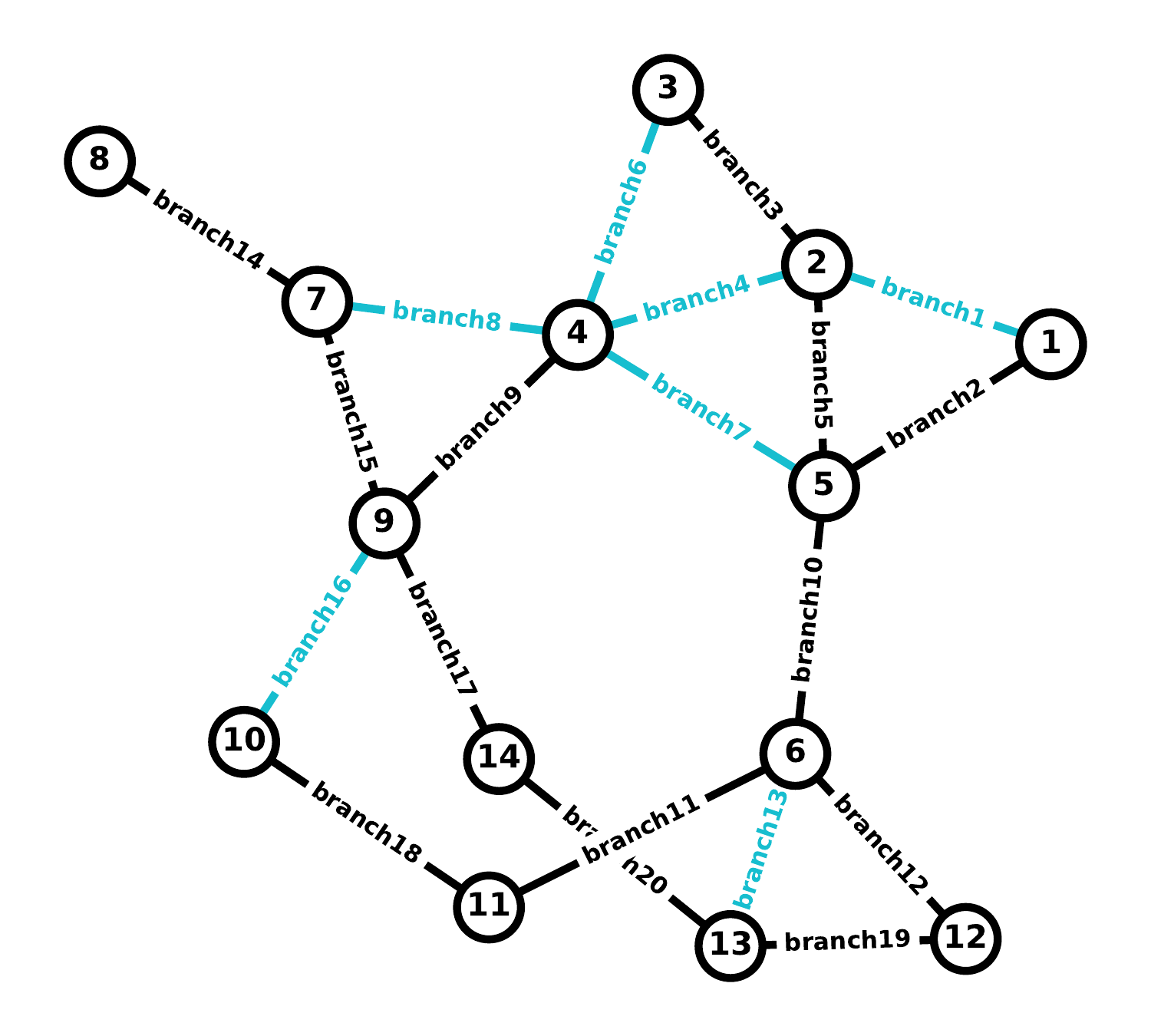}
\caption{The topology of the IEEE 14-Bus system, where the branches selected by BT-MTD ($\Ec_m$) are highlighted in Azure.}
\label{Fig:14BusMTD}
\end{figure}

\subsection{Illustrative Example: IEEE 14-Bus System}

To illustrate Algorithm \ref{Algo:BT-MTD}, the IEEE 14-Bus system in Fig. \ref{Fig:14BusMTD} is used as an example.
The process unfolds in four key topological phases:

\begin{enumerate}
    \item \textbf{Pre-processing (Leaf Pruning):}
    Initially, Bus 8 is identified as a degree-one node. According to Theorem \ref{Theo:deg1}, its incident Branch 14 provides no rank gain and is thus excluded from $\Ec_m$ (marked as covered). Bus 4 is then randomly selected as the traversal root.

    \item \textbf{Root Expansion:}
    From Bus 4, the algorithm selects incident Branches 8, 6, 4, and 7 to maintain connectivity. Branch 9 is skipped as its inclusion renders the remaining subgraph incident to Bus 5 redundant (Theorem \ref{Theo:deg2}).

    \item \textbf{Topological Pruning:}
    The selection in the previous step effectively reduces Buses 3 and 7 to degree-one nodes relative to the remaining graph. Hence, their incident Branches 3 and 15 are deemed redundant and skipped.

    \item \textbf{Cycle-Constrained Selection:}
    Moving to the remaining candidate buses:
    \begin{itemize}
        \item At Bus 13 and 9, Branches 13 and 16 are selected. Branch 17 is skipped as Bus 14 becomes a leaf.
        \item At Bus 2, Branch 5 is rejected because it forms a cycle with the previously selected Branches 4 and 7 (Theorem~\ref{Theo:cycle}). Instead, Branch 1 is selected.
    \end{itemize}
\end{enumerate}

Upon termination, the algorithm yields the optimized set $\Ec_{m} = \{ 1, 4, 6, 7, 8, 13, 16 \}$ with $|\Ec_{m}| = 7$, achieving a final stealthy attack space dimension of 6.

\subsection{ Algorithm Complexity}

The computational complexity of BT-MTD is as follows: 
\paragraph{Pre-processing} Pruning degree-one buses via adjacency list traversal requires $\Oc(|\Vc| + |\Ec|)$ time.
\paragraph{Traversal Loop} 
In each iteration, the algorithm extracts the minimum-degree bus from a priority queue in $\Oc(\log |\Vc|)$ time, sorts its incident branches by neighbor degree in $\Oc(\deg(v_i) \log \deg(v_i))$ time, and tests each incident branch for cycle formation using a Disjoint Set Union (DSU) structure with path compression. 
This DSU operation takes an amortized $\Oc(\alpha(|\Vc|))$ time per operation, where $\alpha$ is the inverse Ackermann function~\cite{tarjan_efficiency_1975}.
Since each branch is examined at most once across all $|\Vc|$ iterations, the aggregate sorting cost is $\Oc(|\Ec| \log \Delta)$ (where $\Delta \leq |\Vc|$ is the maximum nodal degree) and the aggregate DSU cost is $\Oc(|\Ec| \cdot \alpha(|\Vc|))$.

Combining these steps, BT-MTD achieves a quasi-linear complexity of $\Oc(|\Ec| \log |\Vc|)$ by adding multiple branches per iteration.
In contrast, traditional greedy MTD methods select branches individually, requiring expensive $\Oc(n^3)$  SVD/QR rank evaluations per candidate~\cite[Ch.~5, 8]{golub_matrix_2013}, \cite{chen_survey_2025}.
This results in a heavy total complexity of $\Oc(k \cdot |\Ec| \cdot n^3)$.
Consequently, BT-MTD reduces operations by multiple orders of magnitude for large-scale systems, as confirmed in Table \ref{tab:compare_transposed}.

\subsection{Theoretical Performance Evaluation and Robustness}
\label{subsec:theory}

\subsubsection{Theoretical Performance Evaluation}
Recall from Lemma~\ref{Lem:deltaH} that each branch $e_{ij} \in \Ec_m$ contributes specific column vectors to $\Delta\Hm$. 
Hence, the rank of the augmented matrix $[\Hm \mid \Delta\Hm]$ equals the dimension of the vector space spanned by the columns of $\Hm$ together with the columns contributed by $\Ec_m$. 
This corresponds precisely to the rank function of a \textbf{linear matroid} defined on the ground set of column vectors associated with each branch~\cite[Sec.~1.1]{oxley_matroid_2006}.
Since any matroid rank function is monotone and submodular~\cite[Sec.~1.3]{oxley_matroid_2006}, \cite{krause_submodular_2014}, the objective function $\operatorname{rank}([\Hm \mid \Delta\Hm])$ in \eqref{Equ:BiObjectiveOpt} is a \textbf{monotone submodular} set function over subsets $\Ec_m \subseteq \Ec$.


For monotone submodular functions, it is proved in \cite{nemhauser_analysis_1978} that maximizing a monotone submodular function subject to a cardinality constraint (i.e., selecting a set of at most $k$ items) guarantees a solution whose value is at least $(1 - 1/e) \approx 63\%$ times the optimal value.
This implies that BT-MTD can always output a set with at least about $63\%$ of the maximum value.



\subsubsection{Robustness Against Measurement Noise}

It is worth mentioning that the branch selection logic of BT-MTD relies exclusively on \textit{discrete topological properties} (connectivity and cycle formation), which is static or nearly-static and is not affected by the disturbances (such as measurement noise).
To that end, the optimal set $\Ec_{m}$ generated by BT-MTD is inherently decoupled from measurement noise.

\section{Simulation Evaluation}

\subsection{Simulation Setup}
\label{SubSec:SimulationSetup}

To validate the performance of BT-MTD, simulations were performed on standard IEEE test systems using the Python package PYPOWER \cite{zimmerman_matpower_2011} and a MacBook Pro with an Apple M1 Pro processor and 16 GB of RAM.

Following the setting of \cite{morrow_2012_topology} and \cite{zhang_2020_analysis}, the branch admittance modification ratio $\delta_{ij}$ in Lemma \ref{Lem:deltaH} was bounded to the interval $[0.8, 1.2]$ for any branch $ e_{ij} \in \Ec_m$.
This aligns with modern D-FACTS capabilities, such as SmartValve\texttrademark \cite{smartvalve_web}, which can achieve impedance modifications exceeding $\pm20\%$ and enable rapid reconfiguration with ramp time under 200 ms.


\begin{figure*}[t!]
    \centering
    \begin{subfigure}[b]{0.28\textwidth}
        \centering
        \includegraphics[width=\textwidth]{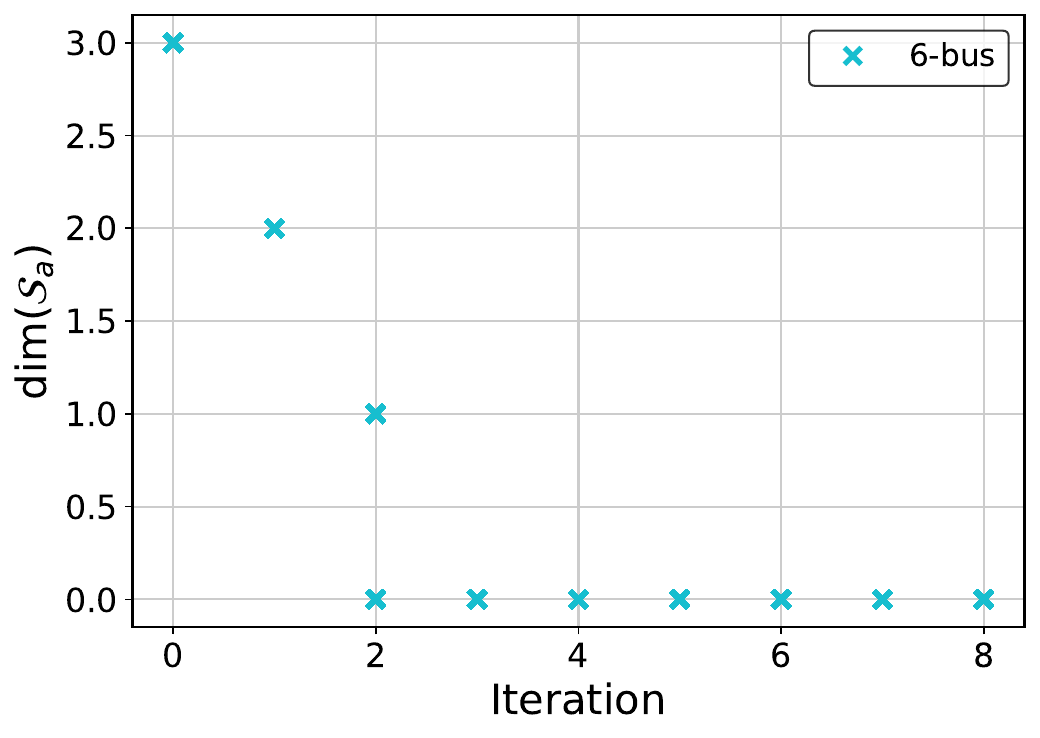}
        \caption{6-bus system}
        \label{fig:6bus}
    \end{subfigure}
    \hfill
    \begin{subfigure}[b]{0.28\textwidth}
        \centering
        \includegraphics[width=\textwidth]{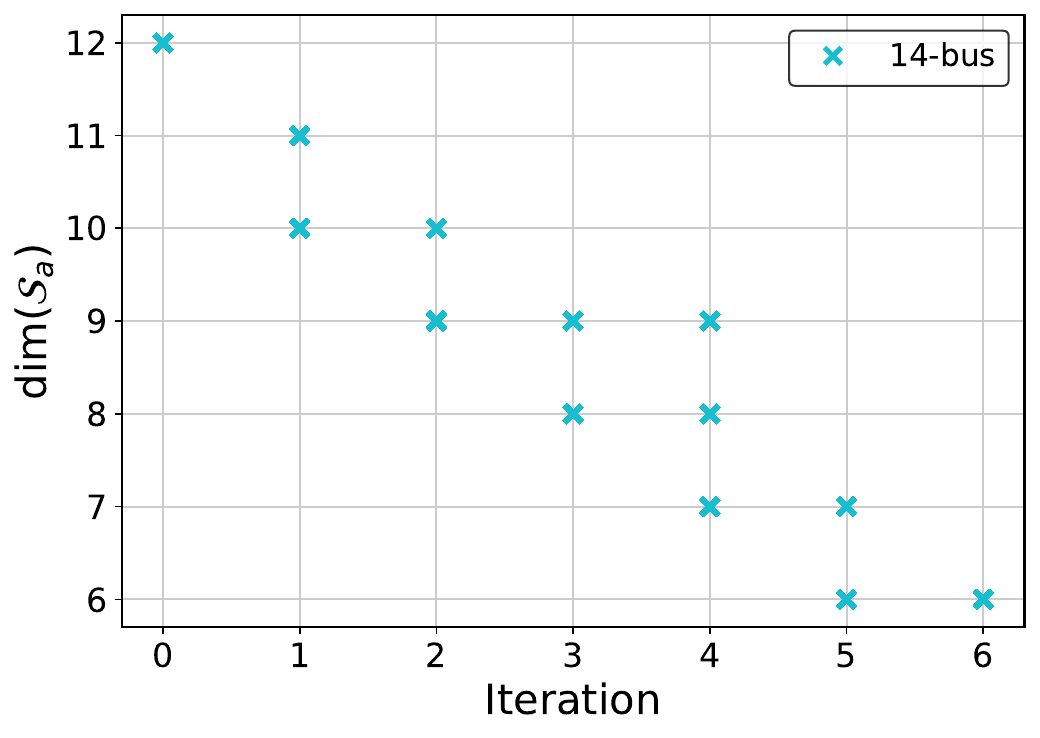}
        \caption{14-bus system}
        \label{fig:14bus}
    \end{subfigure}
    \hfill
    \begin{subfigure}[b]{0.28\textwidth}
        \centering
        \includegraphics[width=\textwidth]{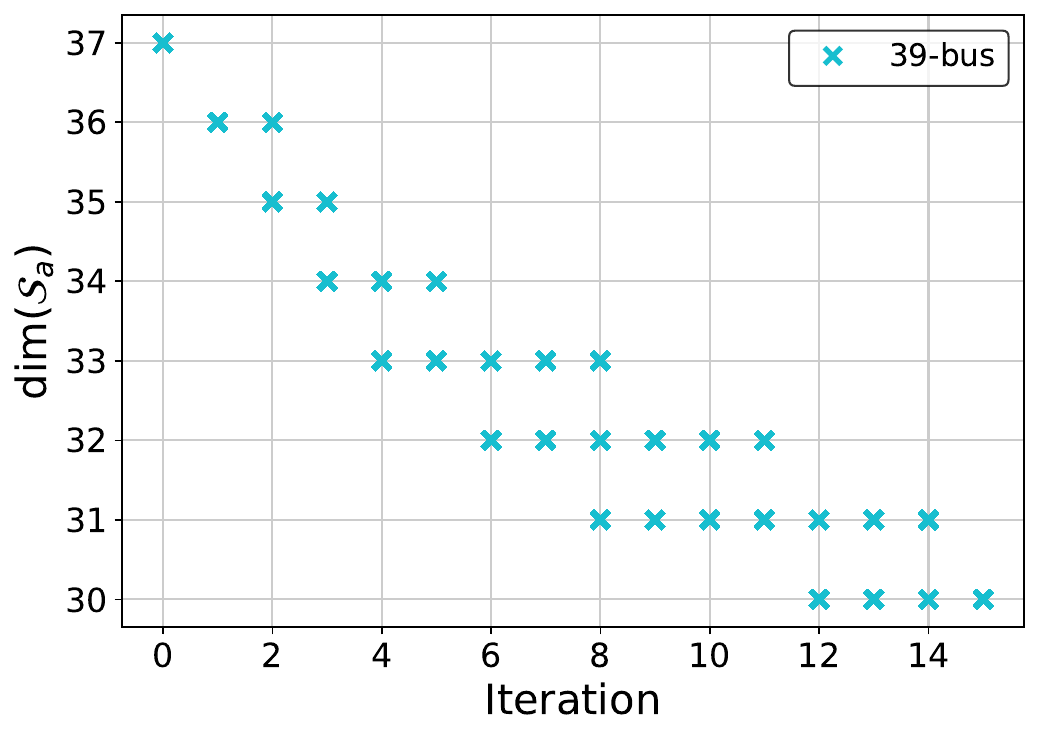}
        \caption{39-bus system}
        \label{fig:39bus}
    \end{subfigure}
    \begin{subfigure}[b]{0.28\textwidth}
        \centering
        \includegraphics[width=\textwidth]{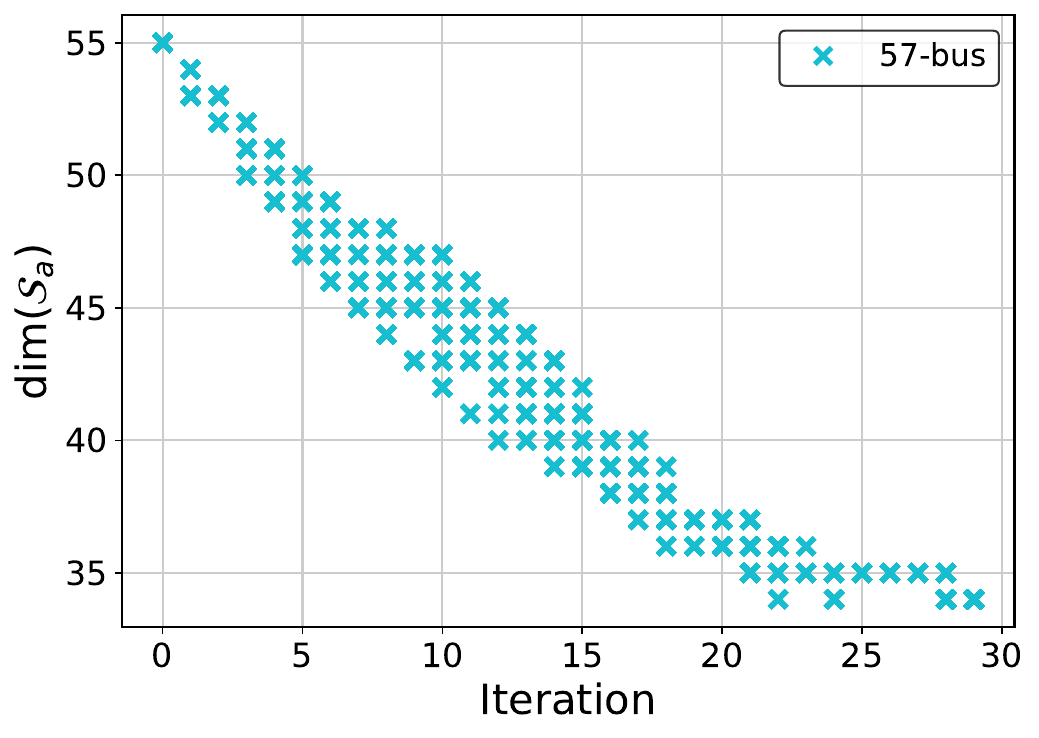}
        \caption{57-bus system}
        \label{fig:57bus}
    \end{subfigure}
    \hfill
    \begin{subfigure}[b]{0.28\textwidth}
        \centering
        \includegraphics[width=\textwidth]{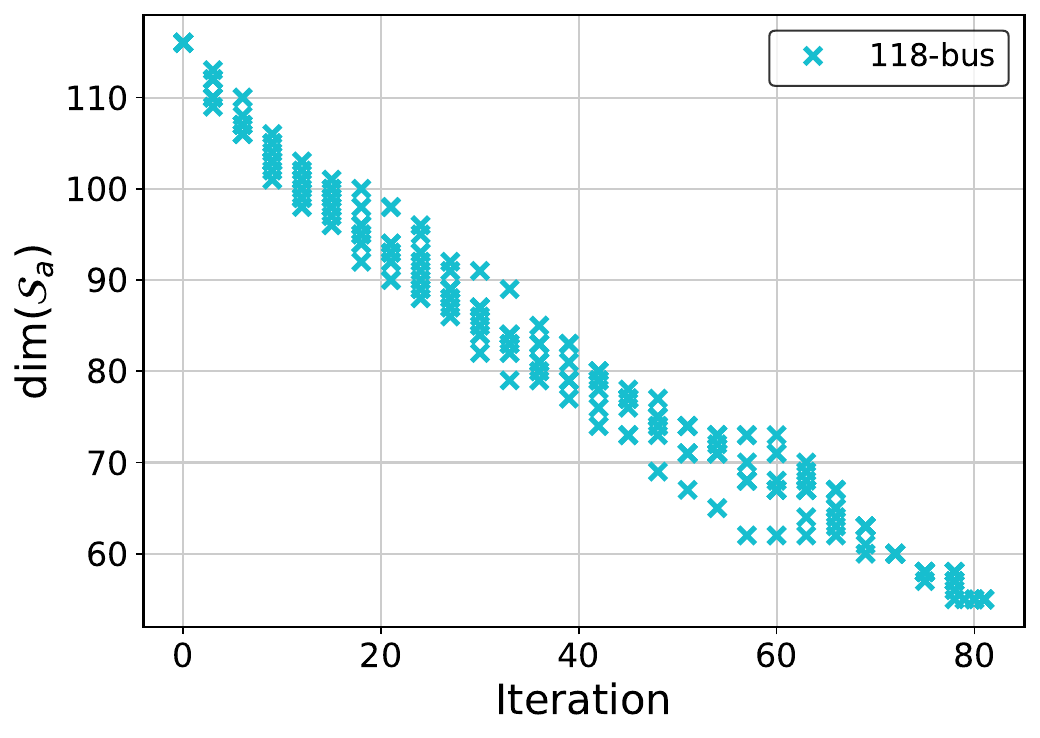}
        \caption{118-bus system}
        \label{fig:118bus}
    \end{subfigure}
    \hfill
    \begin{subfigure}[b]{0.28\textwidth}
        \centering
        \includegraphics[width=\textwidth]{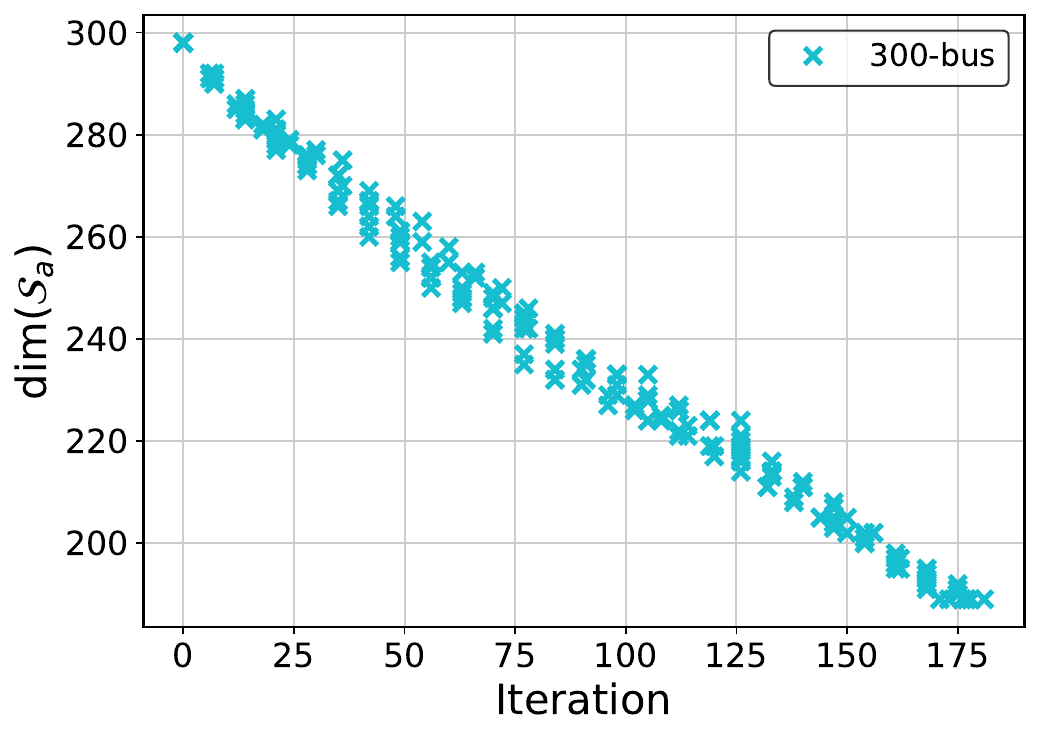}
        \caption{300-bus system}
        \label{fig:300bus}
    \end{subfigure}
    \caption{Performance of the BT-MTD algorithm across different test systems, in which each cross represents the result of one trial at a specific iteration and identical performance from multiple trials at the same iteration is denoted by a single cross.}
    \label{fig:allbuses}
\end{figure*}

\subsection{Performance Evaluation and Robustness Analysis}
\label{subsec:performance_robustness}

The performance of BT-MTD algorithm is visualized in Fig. \ref{fig:allbuses}, where the vertical axis represents the dimension of the stealthy attack space, i.e. $\dim(\mathcal{S}_a)$, and the horizontal axis denotes the number of iterations.
Each iteration corresponds to one execution of the main loop (Lines 8-21 of Algorithm \ref{Algo:BT-MTD}), during which one or more branches are added to the set $\Ec_m$.

\subsubsection{Performance Across Different System Scales}
Fig. \ref{fig:allbuses} demonstrates a consistent trend across all test systems: $\dim(\mathcal{S}_a)$ monotonically decreases as the number of iterations increases.
\textbf{This empirical observation aligns with the theoretical monotone submodularity presented in Section \ref{subsec:theory},} confirming the effectiveness of BT-MTD in identifying impactful branches while avoiding redundancy.
However, the convergence patterns vary by system scale:
\begin{itemize}
    \item \textbf{Small-scale systems} (e.g. 6/14-bus): The dimension of the stealthy attack space reaches the \textbf{converged minimum} within only a few iterations, highlighting the rapid efficacy of BT-MTD on smaller topologies.
    \item \textbf{Large-scale systems} (e.g. 118/300-bus): As the system scale increases, the dimension decreases in a steady pattern before stabilizing. This indicates that for complex grids, BT-MTD maintains a consistent rate of security enhancement per iteration until the topological constraints are met.
\end{itemize}

\subsubsection{Algorithmic Robustness to Initialization}
To assess the impact of the stochastic starting bus selection (Line 6 of Algorithm \ref{Algo:BT-MTD}), BT-MTD was run for 50 trials.
Fig. \ref{fig:allbuses} visualizes the algorithmic robustness: although the trajectories of $\dim(\mathcal{S}_a)$ exhibit vertical dispersion in early iterations due to different traversal paths, this variation diminishes as the algorithm progresses.
Crucially, all trials consistently converge to the same minimum dimension.
This robust convergence substantiates that BT-MTD leverages structural invariants, rather than heuristic guesses, thereby ensuring reliable protection even under randomized deployment scenarios.

\subsubsection{Pareto Optimality under Limited Budget}
As mentioned at the end of Section V-A, varying the number of actuable D-FACTS devices allows BT-MTD to explore the Pareto frontier. Fig.~\ref{Fig:DimSaVsK} illustrates this Pareto frontier for the IEEE 118-bus system, in which the algorithm is run for 30 trails. The plot reveals that the slope is approximately $-1$ dimension per branch.
This linear relationship reflects the optimal forest construction jointly guaranteed by Theorem~\ref{Theo:deg2} and Theorem~\ref{Theo:cycle}.
Specifically, every non-cycle branch commitment contributes exactly one rank dimension, ensuring that the actuation budget is never wasted on redundant branches.

\begin{figure}[t]
\centering
\includegraphics[width=0.6\linewidth]{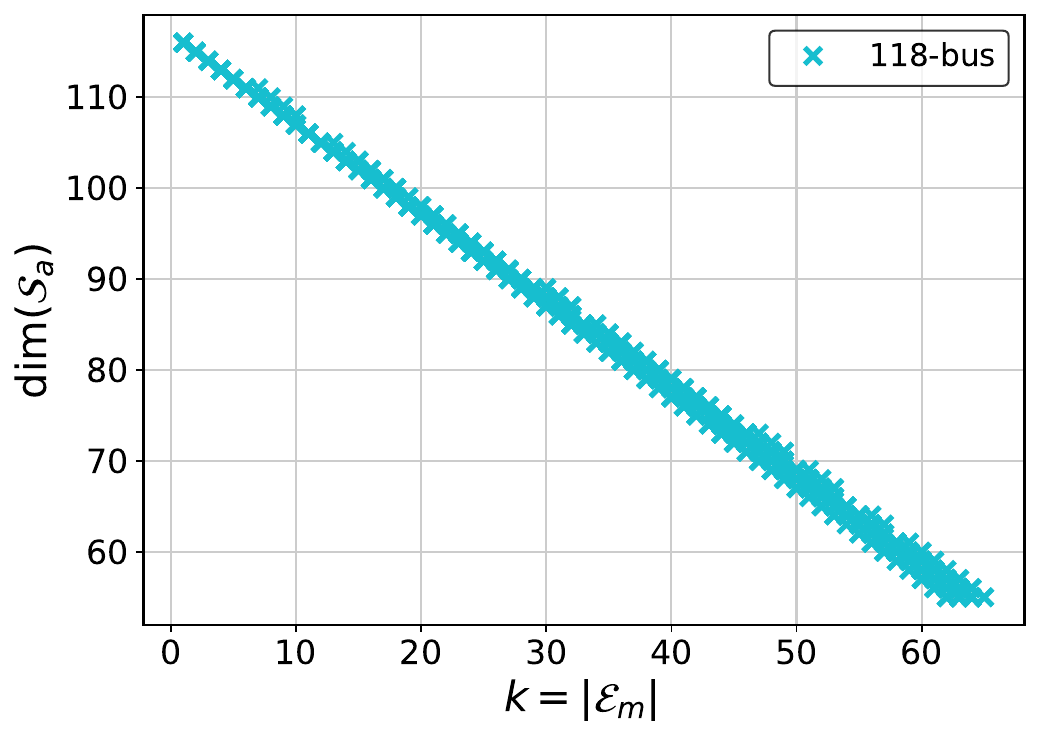}
\caption{Pareto frontier on the IEEE 118-bus system.}
\label{Fig:DimSaVsK}
\end{figure}

\subsection{Comparative Analysis with Existing Strategies}
\label{subsec:comparison}

The proposed BT-MTD algorithm is compared with four established MTD strategies: PFDD \cite{li_feasibility_2020}, Optimal MTD \cite{liu_2021_optimal}, Robust MTD \cite{xu_2023_robust}, and Cyclic-MTD \cite[Algorithm 1]{wang_small-signal_2025}.
The performance of each strategy is evaluated by three metrics:
\textbf{Security Effectiveness}, i.e. the dimension of the final stealthy attack space $\dim({\cal S}_a)$;
\textbf{Resource Efficiency}, i.e. the number of branch admittances that need to be modified, i.e., $\lvert {\cal E}_m \rvert$;
\textbf{Computational Cost}, i.e. the execution time in seconds. 
Note that smaller values for these three metrics indicate better algorithm performance. 
The quantitative results are summarized in Tab. \ref{tab:compare_transposed}.

\begin{table}[t!]
\centering
\caption{Performance Comparison of MTD Strategies}
\label{tab:compare_transposed}
\renewcommand{\arraystretch}{1.2}
\resizebox{\linewidth}{!}{
\begin{tabular}{@{}llcccccc@{}}
\toprule
\textbf{Strategy} & \textbf{Metric} & \textbf{6-bus} & \textbf{14-bus} & \textbf{39-bus} & \textbf{57-bus} & \textbf{118-bus} & \textbf{300-bus} \\
\midrule

\multirow{3}{*}{\makecell[l]{\textbf{PFDD} \\ \scriptsize{(Graph/Cut-set)}}} 
 & $\dim({\cal S}_a)$        & 0      & 6      & 30     & 34     & 55     & 192    \\
 & $\lvert {\cal E}_m \rvert$ & 5      & 13     & 38     & 56     & 117    & 199    \\
 & Time (s)                  & 0.0001 & 0.0002 & 0.0009 & 0.0022 & 0.018  & 0.171  \\
\midrule

\multirow{3}{*}{\makecell[l]{\textbf{Optimal MTD} \\ \scriptsize{(Optimization)}}} 
 & $\dim({\cal S}_a)$        & 0      & 6      & 30     & 34     & 56     & 190    \\
 & $\lvert {\cal E}_m \rvert$ & 7      & 9      & 13     & 26     & 67     & 126    \\
 & Time (s)                  & 0.0029 & 0.0046 & 0.0097 & 0.1186 & 0.1536 & 0.9966 \\
\midrule

\multirow{3}{*}{\makecell[l]{\textbf{Robust MTD} \\ \scriptsize{(Optimization)}}}
 & $\dim({\cal S}_a)$        & 0      & 6      & 30     & 34     & 55     & 189    \\
 & $\lvert {\cal E}_m \rvert$ & 5      & 7      & 14     & 30     & 67     & 126    \\
 & Time (s)                  & 0.001  & 0.0015 & 0.0028 & 0.0085 & 0.043  & 0.249  \\
\midrule

\multirow{3}{*}{\makecell[l]{\textbf{Cyclic-MTD} \\ \scriptsize{(Graph/Cycle)} 
}}
 & $\dim({\cal S}_a)$        & 0     & 6      &30     & 34     & 55     & 189    \\
 & $\lvert {\cal E}_m \rvert$ & 11     & 19     & 35     &77    & 170    & 319    \\
 & Time (s)                 & 0.0005 & 0.0003 & 0.0011 & 0.0121 & 0.0161 & 0.3684\\
\midrule

\multirow{3}{*}{\makecell[l]{\textbf{BT-MTD} \\ \scriptsize{(Graph/Traversal)}}}
 & $\dim({\cal S}_a)$        & 0      & 6      & 30     & 34     & 55     & 189    \\
 & $\lvert {\cal E}_m \rvert$ & 5      & 7      & \bf{8} & \bf{22}& \bf{62}& \bf{114} \\
 & Time (s)                  & 0.0001 & \bf{0.0001} & \bf{0.0005} & \bf{0.0016} & \bf{0.014} & \bf{0.168} \\

\bottomrule
\end{tabular}
}
\end{table}

The results in Tab. \ref{tab:compare_transposed} demonstrate that the proposed BT-MTD algorithm outperforms the other three strategies in nearly every test case.
Only in the smallest system (6-bus system), it achieves the same performance as PFDD and Robust MTD in all three metrics.
In the rest of the test systems, it outperforms the other four strategies in at least one metric, and achieves the same performance in the other metrics.
Especially, its advantages in efficiency and cost become non-negligible in larger test systems, where it consistently requires the fewest branch admittance modifications (smallest value of $\lvert {\cal E}_m \rvert$) and the least computation time, as highlighted by the bolded values.
More importantly, it does not come at the cost of security.
Specifically, BT-MTD provides a final attack space dimension $\dim({\cal S}_a)$ that is either identical to or smaller than the other methods in every scenario.

Furthermore, BT-MTD is qualitatively benchmarked against DD-MTD~\cite{mohammadpourfard_adaptive_2025}.
Although DD-MTD effectively addresses AC nonlinearities through rigorous statistical tests (such as Kullback–Leibler divergence), this high fidelity comes at a substantially higher computational cost.  
By operating at the topological level with $\mathcal{O}(|\mathcal{E}| \log |\mathcal{V}|)$ complexity, 
BT-MTD can significantly reduce the search space for DD-MTD, mitigating its computational burden. 
These complementary approaches suggest a practical two-stage strategy: BT-MTD acts as an efficient topological pre-filter, followed by DD-MTD for high-fidelity parametric refinement. 
Section \ref{Sec:AC} details this tandem deployment.

\section{Discussion}


\subsection{Applicability to AC Models}\label{Sec:AC}

Although the preceding analysis relies on the linearized DC model following standard FDIA formulations \cite{liu_false_2011}, its topological insights remain fundamental for protecting the system under the nonlinear AC model.
Specifically, the detectability of FDIAs depends on whether the attack vector $\av$ lies in the column space of the Jacobian matrix.
Although the AC measurement function is nonlinear, the sparsity patterns of the AC Jacobian ($\Hm_{\text{AC}}$) and DC Jacobian ($\Hm$) are identical, since their nonzero entries are determined solely by the grid’s node–branch connectivity \cite{abur_power_2004}. 
Consequently, the stealthy attack subspaces defined by these structural zeros are the same in both models. 
Therefore,  maximizing $\operatorname{rank}(\Hm)$ (as achieved by BT-MTD) is thus a necessary prerequisite for maximizing $\operatorname{rank}(\Hm_{\text{AC}})$.

Leveraging this structural equivalence, BT-MTD provides a DC-derived structural prior on the candidate set for AC-based MTD solvers.
Its core mechanisms, such as identifying  cycles (Theorem \ref{Theo:cycle}), characterize the incidence-structure observability that is shared by the AC and DC Jacobians.
By operating on topology, BT-MTD efficiently identifies an optimized subset of branches that maximizes the rank with low computational complexity of $\Oc(|\Ec| \log |\Vc|)$.

\subsection{Resilience Against Sophisticated Adversary}
A fundamental assumption of MTD strategy is that the adversary cannot track the admittance changes.
However, a sophisticated adversary could exploit auxiliary information (such as breaker statuses or pre-MTD admittance) and real-time measurements to estimate the current admittances or topology, thereby infer that system parameters have been altered.


Nevertheless, identifying these parameters typically requires far more samples than the time scale of MTD deployment, especially in large networks.
Specifically, the result in \cite{veedu_information_2024} shows that the number of independent samples required to recover the graph structure scales with $\Omega(d \log n)$, where $n$ is the network size and $d$ is the sparsity degree.
For admittance estimation, the result in \cite{rin_admittance_2025} proves that uniquely determining the admittance matrix of an $n$-bus network requires at least $n - 1$ independent phasor measurement snapshots under ideal noiseless conditions.

However, power system measurements exhibit strong temporal correlation due to the quasi-static load profiles, so obtaining sufficiently independent samples often requires observation periods ranging from hours to days \cite{cavraro_learning_2021}.
In contrast, MTD mechanisms commonly rely on fast-acting D-FACTS devices to reconfigure branch admittances within seconds (see Section~\ref{SubSec:SimulationSetup}).
Therefore, the adversary is generally unable to collect sufficient samples to reliably identify the current system admittances or topology.

Another key assumption of the MTD strategy is that the adversary cannot track the specific branches whose admittances have been altered, i.e. the set $\Ec_m$.
If this information is leaked, an attacker can construct a stealth attack that requires no further admittance modifications \cite{sun_2024_stealth}.
However, this attack is only feasible under the necessary condition that the branches in $\Ec_m$
do not form a spanning tree.
As BT-MTD does not perturb a full spanning tree, preventing the leakage of $\Ec_m$ through robust cyber defenses remains a strict and necessary requirement to ensure the security of the MTD deployment.

\section{Conclusion}

In this paper, we proposed the BT-MTD algorithm, a new strategy designed to maximize the effectiveness of MTD with less branch admittance modifications and shorter computational time. 
To achieve this, we identified three key guidelines through theoretic analysis, specifically, avoid branches connecting to degree-one buses, never protect all branches incident to a bus while leaving one unprotected, and ensure no cycle in the set of protected branches. 
Furthermore, a theoretic performance guarantee of BT-MTD is provided via the monotone submodularity of the objective function. 
A comparison with PFDD, Optimal MTD, Robust MTD, and Cyclic-MTD demonstrates that BT-MTD strikes the best balance between effectiveness, efficiency, and computational cost.
Specifically, it achieves a similarly or more effective MTD with less branch admittance modifications and shorter computational time.

\bibliographystyle{IEEEtran}
\bibliography{reference}

\begin{thebibliography}{10}
\providecommand{\url}[1]{#1}
\csname url@samestyle\endcsname
\providecommand{\newblock}{\relax}
\providecommand{\bibinfo}[2]{#2}
\providecommand{\BIBentrySTDinterwordspacing}{\spaceskip=0pt\relax}
\providecommand{\BIBentryALTinterwordstretchfactor}{4}
\providecommand{\BIBentryALTinterwordspacing}{\spaceskip=\fontdimen2\font plus
\BIBentryALTinterwordstretchfactor\fontdimen3\font minus
  \fontdimen4\font\relax}
\providecommand{\BIBforeignlanguage}[2]{{%
\expandafter\ifx\csname l@#1\endcsname\relax
\typeout{** WARNING: IEEEtran.bst: No hyphenation pattern has been}%
\typeout{** loaded for the language `#1'. Using the pattern for}%
\typeout{** the default language instead.}%
\else
\language=\csname l@#1\endcsname
\fi
#2}}
\providecommand{\BIBdecl}{\relax}
\BIBdecl

\bibitem{abdi_role-of-dl_2024}
N.~Abdi, A.~Albaseer, and M.~Abdallah, ``{The Role of Deep Learning in
  Advancing Proactive Cybersecurity Measures for Smart Grid Networks: A
  Survey},'' \emph{IEEE Internet Things J.}, vol.~11, no.~9, pp.
  16\,398--16\,422, May 2024.

\bibitem{meyda_review_2024}
A.~Meydani, H.~Shahinzadeh, A.~Ramezani, H.~Nafisi, and G.~B. Gharehpetian, ``A
  review and analysis of attack and countermeasure approaches for enhancing
  smart grid cybersecurity,'' in \emph{2024 28th International Electrical Power
  Distribution Conference (EPDC)}, Jun. 2024, pp. 1--19.

\bibitem{vasquez-plaza_aggregated_2025}
J.~D. Vasquez-Plaza, S.~Subedi, N.~Bhujel, T.~M. Hansen, R.~Tonkoski, and F.~A.
  Rengifo, ``{Aggregated DER\_A Model Parameterization via Online Moving
  Horizon Estimation},'' \emph{IEEE Trans. Smart Grid}, vol.~16, no.~4, pp.
  3030--3044, Jul. 2025.

\bibitem{sou_resilient_2024}
K.~C. Sou and H.~Sandberg, ``{Resilient Scheduling of Control Software Updates
  in Radial Power Distribution Systems},'' \emph{IEEE Trans. Control Netw.
  Syst.}, vol.~11, no.~3, pp. 1465--1477, Sep. 2024.

\bibitem{ten_vulnerability_2008}
C.-W. Ten, C.-C. Liu, and G.~Manimaran, ``Vulnerability assessment of
  cybersecurity for scada systems,'' \emph{IEEE Trans. Power Syst.}, vol.~23,
  no.~4, pp. 1836--1846, Nov. 2008.

\bibitem{Sun_Stealth_2020}
K.~Sun, I.~Esnaola, S.~Perlaza, and H.~V. Poor, ``Stealth attacks on the smart
  grid,'' \emph{IEEE Trans. Smart Grid}, vol.~11, no.~2, pp. 1276--1285, Mar.
  2020.

\bibitem{mohammadpourfard_adaptive_2025}
M.~Mohammadpourfard, A.~Shefaei, and Y.~Weng, ``An adaptive moving-target
  defense strategy for dynamic nonlinear power systems,'' \emph{IEEE Trans.
  Ind. Inf.}, vol.~21, no.~5, pp. 4136--4145, May 2025.

\bibitem{kececi_federated_2025}
C.~Keçeci, K.~R. Davis, and E.~Serpedin, ``Federated learning-based
  distributed localization of false data injection attacks on smart grids,''
  \emph{IEEE Syst. J.}, vol.~19, no.~3, pp. 719--729, Jul. 2025.

\bibitem{kosut_malicious_2011}
O.~Kosut, L.~Jia, R.~J. Thomas, and L.~Tong, ``Malicious {data} {attacks} on
  the {smart} {grid},'' \emph{IEEE Trans. Smart Grid}, vol.~2, no.~4, pp.
  645--658, Dec. 2011.

\bibitem{liu_matrix-completion_2024}
B.~Liu, H.~Wu, Q.~Yang, H.~Zhang, Y.~Liu, and Y.~Zhang,
  ``Matrix-completion-based false data injection attacks against machine
  learning detectors,'' \emph{IEEE Trans. Smart Grid}, vol.~15, no.~2, pp.
  2146--2163, Mar. 2024.

\bibitem{esmalifalak_bad_2013}
M.~Esmalifalak, G.~Shi, Z.~Han, and L.~Song, ``Bad {data} {injection} {attack}
  and {defense} in {electricity} {market} {using} {game} {theory} {study},''
  \emph{IEEE Trans. Smart Grid}, vol.~4, no.~1, pp. 160--169, Mar. 2013.

\bibitem{seshasai_design_2024}
B.~Seshasai, E.~Koley, P.~K. Jena, and S.~Ghosh, ``Design of real-time false
  data injection attack on electricity market with limited sensor
  accessibility,'' \emph{IEEE Syst. J.}, vol.~18, no.~4, pp. 1999--2009, Oct.
  2024.

\bibitem{bi_profit-oriented_2022}
W.~Bi, G.~Chen, and K.~Zhang, ``Profit-oriented false data injection attack
  against wind farms and countermeasures,'' \emph{IEEE Syst. J.}, vol.~16,
  no.~3, pp. 3700--3710, Sep. 2022.

\bibitem{zhang_voltage-stability_2022}
H.~Zhang, B.~Liu, X.~Liu, A.~Pahwa, and H.~Wu, ``{Voltage Stability Constrained
  Moving Target Defense Against Net Load Redistribution Attacks},'' \emph{IEEE
  Trans. Smart Grid}, vol.~13, no.~5, pp. 3748--3759, Sep. 2022.

\bibitem{liao_cascading_2017}
W.~Liao, S.~Salinas, M.~Li, P.~Li, and K.~A. Loparo, ``Cascading {failure}
  {attacks} in the {power} {system}: {A} {stochastic} {game} {perspective},''
  \emph{IEEE Internet Things J.}, vol.~4, no.~6, pp. 2247--2259, Dec. 2017.

\bibitem{liu_false_2011}
Y.~Liu, P.~Ning, and M.~K. Reiter, ``False {data} {injection} {attacks}
  {against} {state} {estimation} in {electric} {power} {grids},'' \emph{ACM
  Trans. Inf. Syst. Secur.}, vol.~14, no.~1, pp. 13:1--13:33, Jun. 2011.

\bibitem{hug_vulnerability_2012}
G.~Hug and J.~A. Giampapa, ``Vulnerability {assessment} of {AC} {state}
  {estimation} {with} {respect} to {false} {data} {injection}
  {cyber}-{attacks},'' \emph{IEEE Trans. Smart Grid}, vol.~3, no.~3, pp.
  1362--1370, Sep. 2012.

\bibitem{ning_improved_2024}
C.~Ning and Z.~Xi, ``Improved stealthy false data injection attacks in
  networked control systems,'' \emph{IEEE Syst. J.}, vol.~18, no.~1, pp.
  505--515, Jan. 2024.

\bibitem{sun_2023_asymptotic}
K.~Sun, I.~Esnaola, A.~M. Tulino, and H.~V. Poor, ``Asymptotic learning
  requirements for stealth attacks on linearized state estimation,'' \emph{IEEE
  Trans. Smart Grid}, vol.~14, no.~4, pp. 3189 -- 3200, Jul. 2023.

\bibitem{sun_2019_learning}
------, ``Learning requirements for stealth attacks,'' in \emph{Proc. IEEE Int.
  Conf. on Acoust., Speech and Signal Process.}, Brighton, United Kingdom,
  2019, pp. 8102--8106.

\bibitem{lakshminarayana_2020_data}
S.~Lakshminarayana, A.~Kammoun, M.~Debbah, and H.~V. Poor, ``Data-driven false
  data injection attacks against power grid: A random matrix approach,''
  \emph{IEEE Trans. Smart Grid}, vol.~12, no.~1, pp. 635 -- 646, Jan. 2021.

\bibitem{xu_blending_2023}
W.~Xu, M.~Higgins, J.~Wang, I.~M. Jaimoukha, and F.~Teng, ``Blending data and
  physics against false data injection attack: An event-triggered moving target
  defence approach,'' \emph{IEEE Trans. Smart Grid}, vol.~14, no.~4, pp.
  3176--3188, Jul. 2023.

\bibitem{chen_localization_2022}
Y.~Chen, S.~Lakshminarayana, and F.~Teng, ``Localization of coordinated
  cyber-physical attacks in power grids using moving target defense and deep
  learning,'' in \emph{Proc. IEEE Int. Conf. on Smart Grid Commun.}, Oct. 2022,
  pp. 388--393.

\bibitem{cutsem_comprehensive_1995}
T.~Van~Cutsem, Y.~Jacquemart, J.-N. Marquet, and P.~Pruvot, ``A comprehensive
  analysis of mid-term voltage stability,'' \emph{IEEE Trans. Power Syst.},
  vol.~10, no.~3, pp. 1173--1182, 1995.

\bibitem{tan_survey_2023}
J.~Tan, H.~Jin, H.~Zhang \emph{et~al.}, ``A survey: When moving target defense
  meets game theory,'' \emph{Comput. Sci. Rev.}, vol.~48, p. 100544, 2023.

\bibitem{zhang_2012_flexible}
X.~Zhang, C.~Rehtanz, and B.~Pal, \emph{Flexible AC Transmission Systems:
  Modelling and Control}.\hskip 1em plus 0.5em minus 0.4em\relax Springer
  Science \& Business Media, 2012.

\bibitem{morrow_2012_topology}
K.~L. Morrow, E.~Heine, K.~M. Rogers, and T.~J. Bobba, R. B .and~Overbye,
  ``Topology perturbation for detecting malicious data injection,'' in
  \emph{Proc. 45th Hawaii Int. Conf. on System Sciences}, Maui, HI, USA, Jan.
  2012, pp. 2104--2113.

\bibitem{wang_2015_effects}
S.~Wang, W.~Ren, and U.~M. Al-Saggaf, ``Effects of switching network topologies
  on stealthy false data injection attacks against state estimation in power
  networks,'' \emph{IEEE Syst. J.}, vol.~11, no.~4, pp. 2640--2651, Nov. 2015.

\bibitem{MTD}
U.~D. of~Homeland~Security, ``Moving target defense,''
  https://www.dhs.gov/science-and-technology/csd-mtd, Jan. 2023.

\bibitem{zhang_sliding_2023}
G.~Zhang, D.~Tong, Q.~Chen, and W.~Zhou, ``Sliding mode control against false
  data injection attacks in dc microgrid systems,'' \emph{IEEE Syst. J.},
  vol.~17, no.~4, pp. 6159--6168, Jun. 2023.

\bibitem{tan_strategy_2025}
J.~Tan, T.~Zheng, H.~Jin \emph{et~al.}, ``A strategy-making method for {PIoT}
  {PLC} honeypoint defense against attacks based on the time-delay evolutionary
  game,'' \emph{IEEE Trans. Inf. Forensics Security}, 2025.

\bibitem{zhang_2020_analysis}
Z.~Zhang, R.~Deng, D.~K.~Y. Yau, P.~Cheng, and J.~Chen, ``Analysis of moving
  target defense against false data injection attacks on power grid,''
  \emph{IEEE Trans. Inf. Forensics Security}, vol.~15, pp. 2320--2335, Jul.
  2020.

\bibitem{zhang_2020_hiddenness}
------, ``On hiddenness of moving target defense against false data injection
  attacks on power grid,'' \emph{ACM Trans. Cyber-Physical Systems}, vol.~4,
  no.~3, pp. 1--29, Mar. 2020.

\bibitem{xu_2023_robust}
W.~Xu, I.~M. Jaimoukha, and F.~Teng, ``Robust moving target defence against
  false data injection attacks in power grids,'' \emph{IEEE Trans. Inf.
  Forensics Secur.}, vol.~18, pp. 29--40, 2023.

\bibitem{liu_reactance_2018}
C.~Liu, J.~Wu, C.~Long, and D.~Kundur, ``Reactance perturbation for detecting
  and identifying fdi attacks in power system state estimation,'' \emph{IEEE J.
  Sel. Topics Signal Process.}, vol.~12, no.~4, pp. 763--776, Aug. 2018.

\bibitem{tian_enhanced_2019}
J.~Tian, R.~Tan, X.~Guan, and T.~Liu, ``Enhanced hidden moving target defense
  in smart grids,'' \emph{IEEE Trans. on Smart Grid}, vol.~10, no.~2, pp.
  2208--2223, Jan. 2019.

\bibitem{lakshminarayana_2020_cost}
S.~Lakshminarayana and D.~K.~Y. Yau, ``Cost-benefit analysis of moving-target
  defense in power grids,'' \emph{IEEE Trans. Power Systems}, vol.~36, no.~2,
  pp. 1152--1163, Jul. 2020.

\bibitem{zhang_double-benefit_2022}
Z.~Zhang, Y.~Tian, R.~Deng, and J.~Ma, ``A double-benefit moving target defense
  against cyber-physical attacks in smart grid,'' \emph{IEEE Internet Things
  J.}, vol.~9, no.~18, pp. 17\,912--17\,925, Sep. 2022.

\bibitem{wang_MMTD_2023}
J.~Wang, J.~Tian, Y.~Liu, D.~Yang, and T.~Liu, ``{MMTD}: Multistage moving
  target defense for security-enhanced d-facts operation,'' \emph{IEEE Internet
  Things J.}, vol.~10, no.~14, pp. 12\,234--12\,247, Feb. 2023.

\bibitem{davis_power_2012}
K.~R. Davis, K.~L. Morrow, R.~Bobba, and E.~Heine, ``Power flow cyber attacks
  and perturbation-based defense,'' in \emph{Proc. IEEE 3rd Int. Conf. Smart
  Grid Commun.}, Nov. 2012, pp. 342--347.

\bibitem{rahman_2014_moving}
M.~A. Rahman, E.~Al-Shaer, and R.~B. Bobba, ``Moving target defense for
  hardening the security of the power system state estimation,'' in \emph{Proc.
  the First ACM Workshop on Moving Target Defense}, Arizona, Scottsdale, USA,
  Nov 2014, pp. 59--68.

\bibitem{giraldo_decentralized_2022}
J.~Giraldo, M.~El~Hariri, and M.~Parvania, ``Decentralized moving target
  defense for microgrid protection against false-data injection attacks,''
  \emph{IEEE Trans. Smart Grid}, vol.~13, no.~5, pp. 3700--3710, Sep. 2022.

\bibitem{wang_small-signal_2025}
B.~Wang, Z.~Zhang, M.~Wang, M.~Liu, R.~Ye, R.~Deng, and X.~Zhang,
  ``Small-signal-stability-guaranteed moving target defense against load
  redistribution attack on iot-based smart grid,'' \emph{IEEE Internet Things
  J.}, vol.~12, no.~11, pp. 17\,413--17\,429, Jun. 2025.

\bibitem{golub_matrix_2013}
G.~H. Golub and C.~F. Van~Loan, \emph{Matrix Computations}, 4th~ed.\hskip 1em
  plus 0.5em minus 0.4em\relax Baltimore, MD, USA: Johns Hopkins Univ. Press,
  2013.

\bibitem{li_feasibility_2020}
B.~Li, G.~Xiao, R.~Lu, R.~Deng, and H.~Bao, ``On {feasibility} and
  {limitations} of {detecting} {false} {data} {injection} {attacks} on {power}
  {grid} {state} {estimation} using {D-FACTS} {devices},'' \emph{IEEE Trans.
  Ind. Informat.}, vol.~16, no.~2, pp. 854--864, Feb 2020.

\bibitem{lakshminarayana_moving_2021}
S.~Lakshminarayana, E.~V. Belmega, and H.~V. Poor, ``Moving-target defense
  against cyber-physical attacks in power grids via game theory,'' \emph{IEEE
  Trans. Smart Grid}, vol.~12, no.~6, pp. 5244--5257, Jul. 2021.

\bibitem{liu_2021_optimal}
B.~Liu and H.~Wu, ``Optimal planning and operation of hidden moving target
  defense for maximal detection effectiveness,'' \emph{IEEE Trans. Smart Grid},
  vol.~12, no.~5, pp. 4447--4459, Sep. 2021.

\bibitem{wang_stealthiness_2025}
J.~Wang, J.~Tian, G.~Xiao, Y.~Liu, H.~Huang, Y.~Zhou, and T.~Liu, ``On
  stealthiness and effectiveness of moving target defense in smart grids,''
  \emph{IIEEE Trans. Ind. Inf.}, vol.~21, no.~4, pp. 2987--2996, Jan. 2025.

\bibitem{grainger_power_1994}
J.~J. Grainger and W.~D. Stevenson, \emph{Power {S}ystem {A}nalysis}.\hskip 1em
  plus 0.5em minus 0.4em\relax McGraw-Hill, 1994.

\bibitem{abur_power_2004}
A.~Abur and A.~G. Exp{\'o}sito, \emph{Power {System} {State} {Estimation}:
  {Theory} and {Implementation}}.\hskip 1em plus 0.5em minus 0.4em\relax CRC
  Press, Mar. 2004.

\bibitem{sun_2024_stealth}
K.~Sun, I.~Esnaola, and H.~V. Poor, ``Stealth attacks against moving target
  defense for smart grid,'' \emph{arXiv preprint arXiv:2411.16024}, 2024.

\bibitem{kay_1993_fundamentals}
S.~M. Kay, \emph{Fundamentals of Statistical Signal Processing: Estimation
  Theory}.\hskip 1em plus 0.5em minus 0.4em\relax Prentice-Hall, Inc., 1993.

\bibitem{rogers_2008_some}
K.~M. Rogers and T.~J. Overbye, ``Some applications of distributed flexible ac
  transmission system {(D-FACTS)} devices in power systems,'' in \emph{Proc.
  40th North American Power Symp.}, Calgary, Canada, Sep. 2008, pp. 1--8.

\bibitem{liu_2020_optimal}
B.~Liu and H.~Wu, ``Optimal d-facts placement in moving target defense against
  false data injection attacks,'' \emph{IEEE Trans. Smart Grid}, vol.~11,
  no.~5, pp. 4345--4357, Sep. 2020.

\bibitem{tarjan_efficiency_1975}
R.~E. Tarjan, ``Efficiency of a good but not linear set union algorithm,''
  \emph{J. ACM}, vol.~22, no.~2, pp. 215--225, Apr. 1975.

\bibitem{chen_survey_2025}
S.~Lakshminarayana, Y.~Chen, C.~Konstantinou, D.~Mashima, and A.~K. Srivastava,
  ``Survey of moving target defense in power grids: Design principles,
  tradeoffs, and future directions,'' \emph{IEEE Open Access J. Power Energy},
  vol.~12, pp. 455--469, January 2025.

\bibitem{oxley_matroid_2006}
J.~Oxley, \emph{Matroid Theory}.\hskip 1em plus 0.5em minus 0.4em\relax Oxford,
  UK: Oxford University Press, 2006.

\bibitem{krause_submodular_2014}
A.~Krause and D.~Golovin, ``Submodular function maximization,'' in
  \emph{Tractability: Practical Approaches to Hard Problems}, L.~Bordeaux,
  Y.~Hamadi, and P.~Kohli, Eds.\hskip 1em plus 0.5em minus 0.4em\relax
  Cambridge, U.K.: Cambridge University Press, 2014, pp. 71--104.

\bibitem{nemhauser_analysis_1978}
G.~L. Nemhauser, L.~A. Wolsey, and M.~L. Fisher, ``An analysis of
  approximations for maximizing submodular set functions—i,''
  \emph{Mathematical Programming}, vol.~14, no.~1, pp. 265--294, Feb. 1978.

\bibitem{zimmerman_matpower_2011}
R.~D. Zimmerman, C.~E. Murillo-S\'{a}nchez, and R.~J. Thomas, ``{MATPOWER}:
  {Steady}-{state} {operations}, {planning}, and {analysis} {tools} for {power}
  {systems} {research} and {education},'' \emph{IEEE Trans. Power Syst.},
  vol.~26, no.~1, pp. 12--19, Feb. 2011.

\bibitem{smartvalve_web}
\BIBentryALTinterwordspacing
{Smart Wires Inc.} (2024) Smartvalve\texttrademark: Advanced power flow control
  solution. Smart Wires Inc. Accessed: 2026-01-15. [Online]. Available:
  \url{https://www.smartwires.com/smartvalve/}
\BIBentrySTDinterwordspacing

\bibitem{veedu_information_2024}
M.~S. Veedu, D.~Deka, and M.~Salapaka, ``Information theoretically optimal
  sample complexity of learning dynamical directed acyclic graphs,'' in
  \emph{Proc. 27th Int. Conf. Artif. Intell. Statist. (AISTATS)}, ser. Proc.
  Mach. Learn. Res., vol. 238.\hskip 1em plus 0.5em minus 0.4em\relax PMLR, May
  2024, pp. 4636--4644.

\bibitem{rin_admittance_2025}
N.~Rin, I.~Shames, I.~R. Petersen, and E.~L. Ratnam, ``Electric grid topology
  and admittance estimation: Quantifying phasor-based measurement
  requirements,'' in \emph{2025 American Control Conference (ACC)}, Jul 2025,
  pp. 2122--2127.

\bibitem{cavraro_learning_2021}
G.~Cavraro, V.~Kekatos, L.~Zhang, and G.~B. Giannakis, ``Learning power grid
  topologies,'' in \emph{Advanced Data Analytics for Power Systems}, A.~Tajer,
  S.~M. Perlaza, and H.~V. Poor, Eds.\hskip 1em plus 0.5em minus 0.4em\relax
  Cambridge, U.K.: Cambridge Univ. Press, 2021, pp. 3--27.

\bibitem{seber_matrix_2008}
G.~A.~F. Seber, \emph{A {Matrix} {Handbook} for {Statisticians}}.\hskip 1em
  plus 0.5em minus 0.4em\relax John Wiley \& Sons, 2008.

\end{thebibliography}

\appendices
\section{Proof of Lemma \ref{Lem:colH}}
\label{Pro:Lem:colH}
The i-th column of the Jacobian matrix $\hv^i$ is defined by the sum over all branches in the set $\Ec$:
\begin{IEEEeqnarray}{c}
\hv^{i} = \sum_{e_{ij} \in \Ec} b_{ij} \hv_{e_{ij}}^{i}.
\end{IEEEeqnarray}
This summation can be partitioned based on branches incident to bus $\Ec_{v_i}$ versus all other branches $\Ec \setminus \Ec_{v_i}$:
\begin{IEEEeqnarray}{c}
\label{eq:Lem:colH}
\hv^{i} = \sum_{e_{ij} \in \Ec_{v_i}} b_{ij} \hv_{e_{ij}}^{i} + \sum_{e_{ij} \in { \Ec \setminus \Ec_{v_i} }} b_{ij} \hv_{e_{ij}}^{i}. 
\end{IEEEeqnarray}
As established in Lemma \ref{Lem:structureH}, the vector $\hv_{e_{ij}}^i$ is null for any branch $e_{ij}$ not incident to bus $v_i$. 
Consequently, every component in the second summation of \eqref{eq:Lem:colH} is zero, causing the entire sum to vanish.
The expression for $\hv^i$ therefore simplifies to the sum over only the incident branches:
\begin{IEEEeqnarray}{c}
\hv^{i} = \sum_{e_{ij} \in \Ec_{v_i}} b_{ij} \hv_{e_{ij}}^{i}.
\end{IEEEeqnarray}


\section{Proof of Lemma \ref{Lem:comH}}
\label{Pro:Lem:comH}

From Lemma \ref{Lem:structureH}, for any $e_{ij} \in \Ec$, it holds for matrix $b_{ij} \Hm_{e_{ij}}$ that the sum of its columns is a zero vector, i.e. 
\begin{IEEEeqnarray}{c}\label{Equ:add8}
\sum_{k=1}^{n+1} b_{ij} \hv_{e_{ij}}^k = b_{ij} \hv_{e_{ij}}^i + b_{ij} \hv_{e_{ij}}^j = \zerov.
\end{IEEEeqnarray}
Hence, from \eqref{Equ:add1}, it holds for $\Hm$ that 
\begin{IEEEeqnarray}{c}\label{Equ:add9}
\sum_{k=1}^{n+1} \hv^{k} = \zerov,
\end{IEEEeqnarray}
which implies that any column of $\Hm$ can be expressed as a linear combination of the others. 
To that end, removing any column from $\Hm$ does not change the span of the rest columns, which suggests that 
\begin{IEEEeqnarray}{c}\label{Equ:add2}
\operatorname{col}(\Hm) = \operatorname{col}(\bar{\Hm}).
\end{IEEEeqnarray}
Note that \eqref{Equ:add8} and \eqref{Equ:add9} hold for any branch with any admittance value. 
Then, under the same MTD strategy, it holds that 
\begin{IEEEeqnarray}{c}
\label{Equ:add3}
\operatorname{col}(\Hm^{\prime}) = \operatorname{col}(\bar{\Hm}^{\prime}).
\end{IEEEeqnarray}
Then combining \eqref{Equ:add2} with \eqref{Equ:add3} yields the equality in \eqref{Equ:comH}.


\section{Proof of Lemma \ref{Lem:deltaH}}
\label{Pro:Lem:deltaH}

Note that the set $\Ec$ can be decomposed into two parts: the set of branches with modified admittance $\Ec_m$ and the set of branches with unmodified admittance $\{ \Ec \setminus \Ec_{m} \}$.
Hence, from Lemma \ref{Lem:structureH}, it holds for the Jacobian matrix after MTD that 
\begin{IEEEeqnarray}{rCl}
\Hm^{\prime} 
& = &  \sum_{e_{ij} \in \Ec_{m}} \!\! \delta_{ij} b_{ij} \Hm_{e_{ij}} + \!\!\!\!\!\!\!\! \sum_{e_{ij} \in \{ \Ec \setminus \Ec_{m} \} }  \!\!\!\!\!\!\!\! b_{ij}   \Hm_{e_{ij}} \\
& = & \!\! \sum_{e_{ij} \in \Ec_{m}} \!\!\! (\delta_{ij}  \! - \!  1 ) b_{ij} \Hm_{e_{ij}} \!\! + \!\!\! \sum_{e_{ij} \in \Ec_{m}} \!\!\! b_{ij} \Hm_{e_{ij}} \!\! + \!\!\!\!\!\!\!\! \sum_{e_{ij} \in \{ \Ec \setminus \Ec_{m} \} }  \!\!\!\!\!\!\!\! b_{ij}   \Hm_{e_{ij}} \IEEEeqnarraynumspace\\
& = & \!\! \sum_{e_{ij} \in \Ec_{m}} \!\!\! (\delta_{ij} \! - \! 1 ) b_{ij} \Hm_{e_{ij}} \! + \Hm .
\end{IEEEeqnarray}


\section{Proof of Lemma \ref{Lem:ObjFunMax}}
\label{Pro:OptPro}
It holds for the stealthy attack space $\Sc_a$ that 
\begin{IEEEeqnarray}{rCl}
\dim({\cal S}_a)
& = & \dim(\operatorname{col}(\Hm) \cap \operatorname{col}(\Hm^{\prime})) \label{Equ:add4}\\
& = & \operatorname{rank}(\Hm) + \operatorname{rank}(\Hm^{\prime}) - \operatorname{rank}\left(\begin{bmatrix} \Hm \mid \Hm^{\prime} \end{bmatrix}\right) \label{Equ:add5} \IEEEeqnarraynumspace\\
& = & n + n -  \operatorname{rank}\left(\begin{bmatrix} \Hm \mid \Hm^{\prime} \end{bmatrix}\right) \label{eq:rank_sub} \\
& = & 2n - \operatorname{rank}\left(\begin{bmatrix} \Hm \mid \Delta\Hm \end{bmatrix}\right),\label{Equ:add10}
\label{eq:final_identity}
\end{IEEEeqnarray}
where \eqref{Equ:add4} follows from Lemma \ref{Lem:comH}; \eqref{Equ:add5} follows from \cite[3.23(a)]{seber_matrix_2008}; 
\eqref{eq:rank_sub} follows from the fact that the rank of $\Hm$ and $\Hm^{\prime}$ equals to $n$; 
and \eqref{Equ:add10} follows from the fact that subtracting  
the $i$-th column of $\Hm^{\prime}$ by the $i$-th column of $\Hm$ for all $i \in \{1, \ldots, n\}$ yields $\left[ \Hm \mid \Delta\Hm \right]$, and the fact that elementary transformation of a matrix does not change the rank of it. 


\section{Proof of Theorem \ref{Theo:deg1}}
\label{Pro:Theo:deg1}
For set $\Ec_m \cup \Ec_{add}$, it holds for the mismatch $\Delta\Hm^{\prime}$ that 
\begin{IEEEeqnarray}{c}
\Delta\Hm^{\prime} = \Delta\Hm +  \sum_{e_{ij} \in \Ec_{add}} (\delta_{ij} - 1) b_{ij} \Hm_{e_{ij}} = \Delta\Hm + \Delta \Hm_{add}. \squeezeequ \IEEEeqnarraynumspace
\end{IEEEeqnarray}
To that end, to prove the lemma, we only prove that any column of $\Delta \Hm_{add}$ can be linearly represented by the columns of $\Hm$. 
More specifically, we only need to prove that any column of $(\delta_{ij} - 1) b_{ij} \Hm_{e_{ij}}$ with $e_{ij} \in \Ec_{add}$ can be linearly represented by the columns of $\Hm$.

For any branch $e_{ij} \in \Ec_{leaf}$ with $v_i$ being the degree-one bus, it follows from Lemma \ref{Lem:colH} that 
the $i$-th column $\hv^{i}$ of the full Jacobian $\Hm$ is determined solely by that single branch and is given by 
\begin{equation}\label{Equ:rain1}
\hv^{i} = b_{ij}\hv_{e_{ij}}^i.
\end{equation}
And it follows from Lemma \ref{Lem:structureH} that the $i$-th column of $(\delta_{ij} - 1) b_{ij} \Hm_{e_{ij}}$ is $(\delta_{ij} - 1) b_{ij}\hv_{e_{ij}}^i$. 
Hence, the $i$-th column of $(\delta_{ij} - 1) b_{ij} \Hm_{e_{ij}}$ is linearly dependent of the $i$-th column of $\Hm$. 
Adding the fact that the $j$-th column of $(\delta_{ij} - 1) b_{ij} \Hm_{e_{ij}}$ is the negative of the $i$-th column (Lemma \ref{Lem:structureH}), it holds that the $j$-th column of $(\delta_{ij} - 1) b_{ij} \Hm_{e_{ij}}$ is also linearly dependent of the $i$-th column of $\Hm$.

\section{Proof of Theorem \ref{Theo:deg2}}
\label{Pro:Theo:deg2}
When the admittances of branches within $\Ec_m$ and $\Ec_m \cup \Ec_{add}$ are modified, the post-MTD Jacobian matrices are given by $\Hm_1 \eqdef \Hm+ \Delta \Hm$ and $\Hm_2 \eqdef  \Hm+ \Delta \Hm'$, respectively, i.e. 
\begin{IEEEeqnarray}{c}
\Hm \xrightarrow{\Ec_m} \Hm_1 \xrightarrow{\Ec_{add}} \Hm_2
\end{IEEEeqnarray}
Let $b_{ij}^{\prime}$ be the post-MTD branch admittance.
To that end, the difference between these two Jacobian matrices is 
\begin{IEEEeqnarray}{c}
\label{Equ:tansform}
\Delta\Hm_{add} = \sum_{e_{ij} \in \Ec_{add}} (\delta_{ij} - 1) b_{ij}^{\prime} \Hm_{e_{ij}},
\squeezeequ \IEEEeqnarraynumspace
\end{IEEEeqnarray}

Hence, it follows from Lemma \ref{Lem:colH} and Lemma \ref{Lem:deltaH} that $\Delta \Hm_{add}$ contains $\lvert \Ec_{add} \rvert + 1$ non-zero columns, which include $\lvert \Ec_{add} \rvert$ columns defined by
\begin{IEEEeqnarray}{c}\label{Equ:HA1}
\Delta \hv^j_{add} = ( \delta_{ij} \! - \! 1 ) b_{ij}^{\prime} \hv_{e_{ij}}^j
\end{IEEEeqnarray}
and one column defined as
\begin{IEEEeqnarray}{rCl}
\Delta \hv^i_{add} 
& = & \!\! \sum_{e_{ij} \in \Ec_{add}} \!\! ( \delta_{ij} \! - \! 1 ) b_{ij}^{\prime} \hv_{e_{ij}}^i = \!\! \sum_{e_{ij} \in \Ec_{add}} \!\! - ( \delta_{ij} \! - \! 1 ) b_{ij}^{\prime} \hv_{e_{ij}}^j \IEEEeqnarraynumspace \squeezeequ
\end{IEEEeqnarray}
with $\Delta \hv^i_{add}$ being a linear combination of the other $\lvert \Ec_{add} \rvert$ non-zero columns.

Consider the case in which $\Ec_{add} \subset \Ec_{v_i}^{\prime\prime}$ first.
Note that the degrees of the buses incident to bus $v_i$ are larger than one in the reduced graph, which implies that the corresponding columns of $\Hm_1$ are determined by 
\begin{IEEEeqnarray}{c}
\hv_1^{j} =  \sum_{e_{ij} \in \Ec_{v_j}'} \!\!\! b_{ij}' \hv_{e_{ij}}^{j}
\end{IEEEeqnarray}
with $\deg(v_j) >1 $. 
As a result, the $\lvert \Ec_{add} \rvert$ columns of the form in \eqref{Equ:HA1} are linearly independent of the columns in $\Hm_1$, which suggests that 
\begin{IEEEeqnarray}{rl}
&\operatorname{rank}\left(\left[\Hm \mid \Delta \Hm^\prime \right]\right) = 
\operatorname{rank}\left(\left[\Hm \mid \Hm_2 \right]\right) = \operatorname{rank}\left(\left[\Hm \mid \Hm_1 +\Delta\Hm_{add} \right]\right)  \IEEEnonumber \supersqueezeequ \\
&\quad = \operatorname{rank}\left(\left[\Hm \mid \Hm_1 \right]\right) + \lvert \Ec_{add} \rvert = \operatorname{rank}\left(\left[\Hm \mid \Delta \Hm \right]\right)+ \lvert \Ec_{add} \rvert.  \IEEEeqnarraynumspace  \supersqueezeequ
\end{IEEEeqnarray}

Then consider the case in which $\Ec_{add} = \Ec_{v_i}^{\prime\prime}$. 
The relation in \eqref{Equ:tansform} can be formulated as
\begin{IEEEeqnarray}{c}
\Hm \xrightarrow{\Ec_m} \Hm_1 \xrightarrow{\Ec_{add} = \Ec_{v_i}^{\prime\prime}} \Hm_2
\end{IEEEeqnarray}
For this case, it holds that 
\begin{IEEEeqnarray}{c}
\Ec_m \cup \Ec_{add} = \Ec_m \cup \Ec_{v_i}^{\prime}.
\end{IEEEeqnarray}
To that end, the MTD strategy that changes the admittance of branches in $\Ec_m \cup \Ec_{add}$ can be considered as: the admittance of branches in $\Ec_m$ is changed first, then the admittance of branches in $\Ec_{v_i}^{\prime}$ is further changed, in which the admittance of branches in $ \Ec_m \cap \Ec_{v_i}^{\prime}$ are changed twice. 
Then the case in which $\Ec_{add} = \Ec_{v_i}^{\prime\prime}$ can be considered as a case in which $\Ec_{add} = \Ec_{v_i}^{\prime}$, specifically, 
\begin{IEEEeqnarray}{c}
\Hm \xrightarrow{\Ec_m} \Hm_1 \xrightarrow{\Ec_{add} = \Ec_{v_i}^{\prime}} \Hm_2.
\end{IEEEeqnarray}

Let $\bar\Ec_{v_i}$ denote the set of branches removed from $\Ec_{v_i}$ during the graph reduction procedure, i.e. 
\begin{IEEEeqnarray}{c}
\bar\Ec_{v_i} = \{e_{ij} \in \Ec_{v_i} \mid \deg(v_j) = 1 \}.
\end{IEEEeqnarray}
Then it follows from Lemma \ref{Lem:colH} that the $i$-th column of $\Hm_1$ is given by 
\begin{IEEEeqnarray}{rCl}
\hv_{1}^i
& = & 
\!\! \sum_{e_{ij} \in \Ec_{v_i}^\prime} \!\!\! b_{ij}' {\hv}_{e_{ij}}^i +  \!\! \sum_{e_{ij} \in \bar\Ec_{v_i} }  \!\!\! b_{ij}' {\hv}_{e_{ij}}^i \\
& = & \!\!  \sum_{e_{ij} \in \Ec_{v_i}^\prime}  - b_{ij}' {\hv}_{e_{ij}}^j + \sum_{e_{ij} \in \bar\Ec_{v_i} } -  b_{ij}' {\hv}_{e_{ij}}^j \label{Equ:HA2}\\
& = & \!\! \sum_{e_{ij} \in \Ec_{v_i}^\prime} \!\!\! - \left( \Delta \hv_{add}^j \, / \, \left( \delta_{ij} \! - \! 1 \right) \right) - \!\!\! \sum_{e_{ij} \in \bar\Ec_{v_i} } \!\!\! b_{ij}' {\hv}_{e_{ij}}^j \label{Equ:HA3}
\end{IEEEeqnarray}
where \eqref{Equ:HA2} follows from Lemma \ref{Lem:structureH}; \eqref{Equ:HA3} follows from \eqref{Equ:HA1}. 
As a result, it holds that $\hv_1^i$ is a linear combination of $\lvert \Ec_{add} \rvert$ columns of $\Delta \Hm_{add}$ (excluding the $i$-th column) and some columns of $\Hm_1$. 
And only the remaining $\lvert \Ec_{add} \rvert - 1$ columns in $\Delta \Hm_{add}$ contribute to the rank increase, 
which suggests
\begin{IEEEeqnarray}{rl}
&\operatorname{rank}\left(\left[\Hm \mid \Delta \Hm^\prime \right]\right) = 
\operatorname{rank}\left(\left[\Hm \mid \Hm_2 \right]\right) = \operatorname{rank}\left(\left[\Hm \mid \Hm_1 +\Delta\Hm_{add} \right]\right) \IEEEnonumber \supersqueezeequ \\
&\ = \operatorname{rank}\left(\left[\Hm \mid \Hm_1 \right]\right) + \lvert \Ec_{add} \rvert - 1 = \operatorname{rank}\left(\left[\Hm \mid \Delta \Hm \right]\right)+ \lvert \Ec_{add} \rvert - 1.  \IEEEeqnarraynumspace  \Tsupersqueezeequ
\end{IEEEeqnarray}

\section{Proof of Theorem \ref{Theo:cycle}}
\label{Pro:Theo:cycle}
Let $\Delta\Hm$ and $\Delta\Hm'$ being the difference between the post- and pre-MTD Jacobian matrix when the admittance of branches in  $\Ec_m$ and ${\cal E}_{m}^\prime = {\cal E}_{m} \cup \{ e_{ij} \}$ are changed, respectively, i.e. 
This sequence can be represented as:
\begin{IEEEeqnarray}{c}
\Hm \xrightarrow{\Ec_m} \Hm_1 \xrightarrow{\{ e_{ij} \}} \Hm_2
\end{IEEEeqnarray}

Assume that there exists a subset $\Ec_{path} \subseteq \Ec_m$ such that the set $\Ec_{cycle} \eqdef \Ec_{path} \cup \{e_{ij}\}$ forms a cycle.
A central step in this proof is to reframe the MTD strategy with ${\cal E}_{m}^\prime = {\cal E}_{m} \cup \{ e_{ij} \}$ as a two-step sequence, which can be expressed as
\begin{IEEEeqnarray}{c}
\Ec_m \cup \{ e_{ij} \} = \Ec_m \cup (\Ec_{path} \cup \{ e_{ij} \}).
\end{IEEEeqnarray}
Firstly, the admittances of branches in the initial set $\Ec_m$ are modified;
Subsequently, the modifications corresponding to the newly formed cycle $\Ec_{path} \cup \{e_{ij}\}$ is applied, which yields that 
\begin{IEEEeqnarray}{c}
\Hm \xrightarrow{\Ec_m} \Hm_1 \xrightarrow{\Ec_{path} \cup \{ e_{ij} \}} \Hm_2.
\end{IEEEeqnarray}
Note that this implies that the admittances of branches in $\Ec_{path}$ are modified twice.

To that end, we only needs to prove that the inclusion of branch $e_{ij}$, which forms a cycle with branches already in $\Ec_m$, does not introduce any new linearly independent columns into $\Delta\Hm$.
Let $\Delta\Hm_{path}$ and $\Delta\Hm_{cycle}$ be the difference between post- and pre-MTD Jacobian matrix when the admittance of branches in $\Ec_{path}$ and $\Ec_{path} \cup \{e_{ij}\}$ are modified. 
Then from Lemmas \ref{Lem:colH}-\ref{Lem:deltaH}, it holds that $\Delta\Hm_{path}$ has $\lvert\Ec_{path}\rvert + 1$ non-zero columns and a rank of $\operatorname{rank}(\Delta\Hm_{path}) = \lvert\Ec_{path}\rvert$; 
and $\Delta\Hm_{cycle}$ also contains $\lvert\Ec_{path}\rvert + 1$ non-zero columns. 
Note that the sum of the columns in $\Delta\Hm_{cycle}$ is a zero vector, which implies the columns are linearly dependent. 
As a result, it holds that 
\begin{IEEEeqnarray}{c}
\operatorname{rank}(\Delta\Hm_{cycle}) = \lvert\Ec_{path}\rvert = \operatorname{rank}(\Delta\Hm_{path}), \IEEEeqnarraynumspace
\end{IEEEeqnarray}
which implies that 
\begin{IEEEeqnarray}{rl}
& \operatorname{rank}\left(\left[\Hm \mid \Delta \Hm \right]\right) = \operatorname{rank}\left(\left[\Hm \mid \Hm_1 \right]\right) = \operatorname{rank}\left(\left[\Hm \mid \Delta \Hm_{path} \right]\right) = \IEEEeqnarraynumspace \IEEEnonumber \supersqueezeequ\\
&  \  \operatorname{rank}\left(\left[\Hm \mid \Delta \Hm_{cycle}\right]\right)  = \operatorname{rank}\left(\left[\Hm \mid \Hm_2 \right]\right)  = \operatorname{rank}\left(\left[\Hm \mid \Delta \Hm^\prime \right]\right)\IEEEeqnarraynumspace  \supersqueezeequ
\end{IEEEeqnarray}


\end{document}